\documentclass[reqno,10pt]{amsart}
\usepackage{amsmath, amsfonts, amsthm, amssymb, amscd}
\usepackage[mathcal]{euscript} 
\usepackage{dsfont, pxfonts, mathrsfs, accents, verbatim} 
\theoremstyle{plain}
        \newtheorem{theorem}{Theorem}[section]
        \newtheorem{proposition}[theorem]{Proposition} 
        \newtheorem{lemma}[theorem]{Lemma}
        \newtheorem{corollary}[theorem]{Corollary} 
        \newtheorem{definition}[theorem]{Definition} 
        \newtheorem{remark}[theorem]{Remark}  
\numberwithin{equation}{section} 

\newcommand \loc   {\text{loc}}
\newcommand \Ric 		{\textbf{Ric}}  
\newcommand \Riem 		{\textbf{Riem}}  
\newcommand \Tor        {\textbf{T}} 
\newcommand \RR 		{\mathbb{R}}  
\newcommand \del 		\partial
\newcommand \eps 		\epsilon
\newcommand \lam 		\lambda 
\newcommand \be 		{\begin{equation}}
\newcommand \ee 		{\end{equation}}
\newcommand \utR		{{\undertilde{R}}}
\newcommand \Vect      	{{\text{Vect}}}
 
\newcommand \Lcal 	 	{\mathcal L}
\newcommand \Dcal 	 	{\mathcal D}
\newcommand \Ccal 	 	{\mathcal C}
\newcommand \TT 	 	{\mathfrak{T}}
\newcommand \X   	 	{ T_0^1}
\newcommand \XX 	 	{ \TT_0^1}
\newcommand \Y   	 	{ T_1^0}
\newcommand \YY 	 	{ \TT_1^0}

\newcommand \MM 	 	{\mathcal M}
\newcommand \HH 	 	{\mathcal H}
\newcommand \tX  		 {{\widetilde X}\hskip.05cm} 
\newcommand \tl   		{{\widetilde \ell}\hskip.05cm}
\newcommand \tg   		{{\widetilde{g}} \hskip.05cm}
\newcommand \uttheta   	{{\undertilde{\theta}} \hskip.05cm}
\newcommand \utn  		{{\undertilde{n}} \hskip.05cm}
\newcommand \utg   	{{\undertilde{g}} \hskip.05cm}
\newcommand \otg   	{{\widetilde{g}} \hskip.05cm}
\newcommand \utX   	{{\undertilde X}\hskip.05cm} 
\newcommand \utl   		{{\undertilde \ell}\hskip.05cm} 
\newcommand \ttheta   	{{\widetilde{\theta}} \hskip.05cm}
\newcommand \tn  		{{\widetilde{n}} \hskip.05cm} 
 
\newcommand \utnabla  	{{\undertilde{\nabla}} \hskip.05cm} 
\newcommand \otnabla  	{{\widetilde{\nabla}} \hskip.05cm} 
\newcommand \utRiem  	{{\undertilde{\textbf{Riem}}} \hskip.05cm}

\newcommand \ot  		{\widetilde}
\newcommand \ut 		{\undertilde}

\newcommand \la 		{\langle}
\newcommand \ra 		{\rangle} 
\newcommand \lad 		{\prec}
\newcommand \rad 		{\succ_{{\mathcal D}',{\mathcal D}}}
\newcommand \dM 	 	{{m}}
\newcommand \dH 	 	{{m-1}}
\newcommand \Dirac 		{{\boldsymbol \delta}}
\newcommand \Edual 		{E}

\begin{document}
\title[Lorentzian manifolds with distributional curvature]
{Definition and stability of Lorentzian manifolds 
\\
with distributional curvature}
\author
   [P.G. L{\tiny e}Floch \and C. Mardare] 
{Philippe G. L{\tiny e}Floch
 and Cristinel Mardare}
\date{}
\keywords{Lorentzian manifold, connection, curvature, distribution, general relativity.
\newline 
\footnote""{{\em Address :} Laboratoire Jacques-Louis Lions \& Centre National de la Recherche Scientifique, 
Universit\'e de Paris 6, Place Jussieu, 75252 Paris, France. E-mail: 
{\tt LeFloch@ann.jussieu.fr, Mardare@ann.jussieu.fr}
\newline 
Published in : Port. Math. 64 (2007), 535--574.} 
}

\begin{abstract} 
Following Geroch, Traschen, Mars and Senovilla, we consider Lorentzian manifolds with 
distributional curvature tensor. Such manifolds represent spacetimes of general relativity
that possibly contain gravitational waves, shock waves, and other singular patterns. 
We aim here at providing a comprehensive and geometric (i.e., coordinate-free) framework.  
First, we determine the minimal assumptions required on the metric tensor 
in order to give a rigorous meaning to the spacetime curvature within the framework of distribution theory.
This leads us to a direct derivation of the jump relations associated with singular parts of 
connection and curvature operators.  
Second, we investigate the induced geometry on a hypersurface with general signature,  
and we determine the minimal assumptions required to define, in the sense of distributions,
the curvature tensors and the second fundamental form of the hypersurface and to establish 
the Gauss-Codazzi equations. 
\end{abstract}

\maketitle


\section{Introduction}
\label{IN}
 
Our main motivation for a study of Lorentzian manifolds with distributional curvature 
comes from general relativity: a spacetime is a $(3+1)$-dimensional differential manifold 
$\MM$ endowed with a Lorentzian metric $g$ with signature 
$(-, +,+,+)$, satisfying Einstein field equations (in normalized units) 
\be
G_{\mu\nu} = T_{\mu\nu},
\label{IN-Einstein}
\ee
where $G_{\mu\nu} := R_{\mu\nu} - (R/2) g_{\mu\nu}$ is Einstein's curvature tensor, 
$R_{\mu\nu}$ the Ricci curvature, $R$ the scalar curvature, and $T_{\mu\nu}$ 
the stress-energy tensor describing the matter content of the spacetime under consideration. 
Singular spacetimes having metric tensor with limited regularity are of particular importance 
in general relativity; many explicitly known solutions of \eqref{IN-Einstein} exhibit black holes, 
gravitational waves, shock waves, or other singular features.  
For instance, the metric can be smooth everywhere except on a smooth hypersurface ${\HH \subset \MM}$
across which the curvature tensor suffers a jump discontinuity; such a hypersurface is interpreted physically 
as a gravitational wave propagating in the spacetime. 
Recall also that, according to Penrose and Hawking incompleteness theorems, 
spacetimes are sought to be generically singular \cite{Rendall2,Wald}. 

Our aim in the present paper is to investigate the local properties of singular spacetimes 
and of their hypersurfaces, within the theory of distributions. 
Although this issue has already been addressed extensively 
\cite{BarrabesIsrael,ClarkeDray,GerochTraschen,MarsSenovilla,Penrose,Raju,SmollerTemple,Taub},  
it appears that, in the mathematical literature, 
no comprehensive discussion of the minimal assumptions required to give a rigorous 
meaning to the curvature in theory of distributions is currently available. The approach we propose follows 
earlier pioneering work by Geroch and Traschen \cite{GerochTraschen} and by Mars and Senovilla \cite{MarsSenovilla}, 
but especially aims at providing a comprehensive and fully geometric exposition. 
That is, we avoid any reference to specific coordinate charts on the manifold, and we provide 
a direct and natural derivation of singular parts of curvature tensors.   

More precisely, a $\Ccal^\infty$-differentiable $m$-dimensional manifold $\MM$ 
being fixed, we seek for the minimal regularity required on a metric tensor $g$ 
defined on $M$, in order to rigorously define (as distributions) 
the connection operator $\nabla$ and the curvature tensor $\Riem$ associated with this metric. 
The same question arises when a connection $\nabla$ is prescribed on $M$ and we attempt to define its curvature. 
To study the geometric properties of a differentiable manifold endowed with a 
non-smooth metric or connection, the proper functional framework is that of distributions.  
We introduce below several definitions of distributional 
metric, connection, and curvature and, under various assumptions, 
we discuss the (weak or strong) stability properties of sequences of distributional 
metrics, connections, or curvatures. 

Our presentation allows us to derive 
jump relations for the singular parts of these quantities, once they are viewed as distributions. 
In Section~\ref{PG-1} we investigate the situation that a connection is provided on the manifold, 
and in Section~\ref{geometry-metric} we consider the case of a metric tensor. The signature of the 
metric is irrelevant for this first part. 

In a second part, in Sections~\ref{HY-0} and \ref{CU-0}, 
we turn our attention to hypersurfaces $\HH$ within a Lorentzian manifold $\MM$, 
when the prescribed metric (or connection) typically suffers a jump discontinuity across $\HH$. 
Our discussion applies to hypersurfaces with general signature, which are not globally timelike, 
spacelike, or null but may change type from point to point. Hence, the hypersurface can be locally 
Riemannian, Lorentzian, or degenerate, and it is important to carefully distinguish between various 
geometric objects defined in $\MM$ which may, or may not, have traces on the hypersurface.  
We discuss the nature and regularity of the geometry induced on the hypersurface by the geometry
of the ambiant spacetime. 

On one hand, a connection being given in the manifold together with a ``rigging field'' on the hypersurface $\HH$ 
(see Section~\ref{HY-0}), we determine an induced connection on $\HH$, denoted below by $\utnabla$. 
On the other hand, a second concept of induced connection on $\HH$, 
denoted by $\otnabla$, can be defined when the connection $\nabla$ is the Levi-Cevita connection 
of a given metric. We observe that 
the connection $\utnabla$ arises as a more natural concept, as was recognized in \cite{MarsSenovilla}. 

The material presented in the present paper should find applications in several directions. 
One one hand, based on the jump relations derived in this paper, 
one should construct a large class of singular vacuum spacetimes containing impulsive gravitational waves. 
The metrics satisfy here the Einstein equations \eqref{IN-Einstein}
which impose further constrains beyond the geometric ones on the nature of the discontinuities. 
Following the approach in \cite{BLSS,CaciottaNicolo,Rendall}, such spacetimes are obtained by solving 
a characteristic-value problem for the Einstein equations with initial data prescribed on a hypersurface.  
Singular matter spacetimes containing 
gravitational waves and shock waves have been recently also constructed
by solving the Einstein-Euler equations for Gowdy symmetric spacetimes \cite{BLSS,LeFlochStewart}. 

On the other hand, in the context of numerical relativity, the formulation of suitable boundary conditions 
\cite{FriedrichNagy,FriedrichRendall,ReulaSarbach} is an important issue, 
and the analysis in the present paper should be relevant to handle boundary with general signature. 
Recall that various excision methods have been devised in the literature to attempt to cut out of the 
numerical domain the black hole regions which, in principle, should not influence the regular part of 
the spacetime. However, many difficulties arise with such techniques at both the theoretical and the numerical levels, 
and further research is necessary to ensure the nonlinear stability of such numerical methods.

\section{Preliminaries}  
\label{PG-0}

\subsection{Tensors and integration on a differentiable manifold}
\label{basics} 

Throughout this paper, $\MM$ denotes a connected, oriented, $\Ccal^\infty$ differentiable $\dM$-manifold. 
The tangent and cotangent spaces at $x \in \MM$ are denoted by $T_x\MM$ and $T_x^\star\MM$, 
and the corresponding bundles by $T\MM := \bigcup_{x \in M} T_x\MM$ and 
$T^\star\MM := \bigcup_{x \in \MM} T_x^\star \MM$, respectively; 
the action of a covector (or $1$-form) $\omega$ on a vector $X$ is denoted by $\la \omega, X\ra$. 
The bundle of all ($p$-contravariant and $q$-covariant) $(p,q)$-tensors is denoted by 
$T_q^p\MM := \bigcup_{x \in \MM} T_{q,x}^p \MM$, and 
is canonically endowed with a structure of $\Ccal^\infty$-differentiable manifold. 
A $(1,0)$-tensor field is identified with a vector field $X$ on $\MM$, that is, 
$X_x \in T_x\MM$ for all $x\in\MM$. 

We denote by $\Lambda^k(\MM)$ the bundle of differential forms of order $k \leq m$, that is,
$(0,k)$-tensor fields that are anti-symmetric with respect to any pair of variables. 
Clearly, $\Lambda^k(\MM)\subset T_k^0(\MM)$, with 
$\Lambda^0(\MM)=T_0^0(\MM)$ (the space of functions on $\MM$) and $\Lambda^1(\MM)=\Y(\MM)$. 
 
We introduce the space $\Ccal^\infty(\MM)$ consisting of all $\Ccal^\infty$-differentiable functions 
$f:\MM \to \RR$  and, more generally, the space $\Ccal^\infty T_q^p(\MM)$ of all 
$\Ccal^\infty$-sections $S: x \in \MM \mapsto S_x \in T_{q,x}^p\MM$. The following short-hand notation will 
also be used 
$$
\TT_q^p(\MM) := \Ccal^\infty T_q^p(\MM). 
$$
Similarly, the spaces of compactly supported $\Ccal^\infty$-functions, tensor fields, and 
differential forms 
will be denoted by $\Dcal(\MM)$, $\Dcal T_q^p(\MM)$, and $\Dcal\Lambda^k(\MM)$, respectively. 
  
Non-smooth tensor fields will be also useful. We will consider tensor fields in the Lebesgue and Sobolev spaces  
$L^r_\loc T_q^p(\MM)$ and $W^{k,r}_\loc T_q^p(\MM)$ ($k,r\geq 1$)
consisting of all $(p,q)$-tensors whose $r$-powers are locally integrable on $\MM$
or belong to the corresponding Sobolev space of order $k$. 
This regularity can be checked in any system of local coordinates, and 
it is important to realize that, although they are unambiguously defined from the sole $\Ccal^\infty$ differentiable 
structure of the manifold $\MM$, all these spaces of tensor fields {\sl are not} endowed with canonical norms.  

We will also consider sequences of tensor fields. A sequence of tensors $A^{(n)}\in 
W^{k,r}_\loc T_q^p(\MM)$ is said to converge in the strong (weak, respectively) $W^{k,r}_\loc$ topology to 
some limit tensor field $A^{(\infty)}\in W^{k,r}_\loc T_q^p(\MM)$  
if for all $X_{(i)}\in\XX(\MM)$ and $\theta^{(j)}\in\YY(\MM)$,
the sequence of functions 
$$
A^{(n)}(X_{(1)},\ldots,X_{(p)},\theta^{(1)},\ldots,\theta^{(q)})  
\to 
A^{(\infty)}(X_{(1)},\ldots,X_{(p)},\theta^{(1)},\ldots,\theta^{(q)})
$$
converges in the strong (weak, resp.) $W^{k,r}_\loc(\MM)$ topology.
This definition is equivalent to the convergence of the components of $A^{(n)}$ in any chosen coordinate atlas.

Observe also that, since the manifold $\MM$ is oriented, the integral of $\dM$-forms is well-defined on
open sets, and Stokes formula 
\be
\label{Stokes}
\int_{\MM'} d \omega = \int_{\del\MM'} \omega
\ee 
holds for all open set $\MM' \subset \MM$ with smooth boundary $\del M'$ and 
for any $(\dH)$-form $\omega$. Here, $d$ denotes the operator of exterior differentiation. 

Recalling the interior product $i_X \omega$ 
defined for all $k$-forms $\omega$ and vectors $X,Z_{(1)}, \ldots, Z_{(k)}$ by
$$
(i_X\omega)(Z_{(1)}, \ldots, Z_{(k)}) := \omega(X,Z_{(1)}, \ldots, Z_{(k)}), 
$$  
we can express the Lie derivative $\Lcal_X$ of a vector field $X$ as 
$$
\Lcal_X=di_X+i_Xd. 
$$ 
Then, writing 
$$
\aligned 
(Xf) \, \omega 
& = \Lcal_X(f \, \omega) - f \, \Lcal_X\omega
\\
& = d(i_X(f \, \omega)) - f \, \Lcal_X\omega, 
\endaligned 
$$ 
and using \eqref{Stokes} we obtain the formula 
\be
\label{formula-div}
\int_{\MM'} (Xf) \, \omega = \int_{\del \MM'} f \, i_X \omega - \int_{\MM'} f \, \Lcal_X \omega 
\ee
for all smooth vector fields $X$, functions $f$, and $m$-form fields $\omega$. 

\subsection{Distributions on a differentiable manifold}
\label{distributions}

We now introduce the notion of tensor distributions on a manifold. It is convenient to 
define first: 

\begin{definition}
The space of {\rm scalar distributions} $\Dcal'(\MM)$ is the dual 
of the space $\Dcal\Lambda^m(\MM)$ of all compactly supported {\rm densities}. 
The space of {\rm distribution densities} $\Dcal' \Lambda^m(\MM)$ is the dual of the space 
$\Dcal(\MM)$ of compactly supported {\rm functions}.
\end{definition}

The duality bracket between a scalar distribution $A \in \Dcal'(\MM)$ and a density $\omega \in \Dcal\Lambda^\dM(\MM)$ is 
written as $\lad A,\omega\rad$. In view of \eqref{formula-div}, we have 
$$
\int_\MM (Xf)\omega =- \int_\MM f\Lcal_X\omega, 
\qquad f\in\Ccal^\infty(\MM), \quad \omega\in\Dcal\Lambda^\dM(\MM),
$$ 
and 
it is natural to define action $X A$ of a smooth vector field $X\in \XX(\MM)$ on a scalar distribution 
$A \in \Dcal'(\MM)$ by using the Lie derivative, i.e.,    
$$
\lad XA, \omega\rad := - \lad A,\Lcal_X\omega\rad, \qquad \omega\in\Dcal\Lambda^\dM(\MM). 
$$ 

Observe that the space of locally integrable functions is {\sl canonically} 
embedded into the space of scalar distributions, that is, $f \in L^1_\loc(\MM) \mapsto f \in \Dcal'(\MM)$, 
via 
$$
\lad f,\omega\rad:=\int_\MM f\omega, \qquad \omega\in \Dcal\Lambda^\dM(\MM). 
$$

More generally, we define a {\sl $(p,q)$-tensor distribution} as a $\Ccal^\infty(\MM)$-multi-linear map 
$$
A : \underbrace{\XX(\MM)\times \ldots \times \XX(\MM)}_{p\text{ times}} \times 
        \underbrace{\YY(\MM) \times \ldots \times \YY(\MM)}_{q\text{ times}}\to \Dcal'(\MM),
$$
and denote the space of tensor distributions by $\Dcal' T_q^p(\MM)$. 
The space of locally integrable tensor fields $L^1_\loc T_q^p(\MM)$ is {\sl canonically} embedded into the space 
$\Dcal' T_q^p(\MM)$, that is, $A \in L^1_\loc T_q^p(\MM) \mapsto A \in \Dcal' T_q^p(\MM)$ 
via 
$$
\lad A(X_{(1)},\ldots,X_{(p)},\theta^{(1)},\ldots,\theta^{(q)}),\omega\rad 
:= \int_\MM  A(X_{(1)},\ldots,X_{(p)},\theta^{(1)},\ldots,\theta^{(q)}) \, \omega, 
$$
for all $\omega \in \Dcal \Lambda^m(\MM)$, $X_{(1)},\ldots,X_{(p)} \in \XX(\MM)$ and 
$\theta^{(1)},\ldots,\theta^{(q)} \in \YY(\MM)$,

We will also consider limits of sequences of distributions. 
A sequence of $(p,q)$-tensor distributions $A^{(n)}$ is said to 
converge in the distribution sense to a limit $A^{(\infty)}$ if, in $\Dcal'(\MM)$, 
$$
A^{(n)}(X_{(1)},\ldots,X_{(p)},\theta^{(1)},\ldots,\theta^{(q)})
\to 
A^{(\infty)}(X_{(1)},\ldots,X_{(p)},\theta^{(1)},\ldots,\theta^{(q)})  
$$
for all $X_{(i)}\in\XX(\MM)$ and $\theta^{(j)}\in\YY(\MM)$. 
In other words, $A^{(n)}$ converges in the distribution sense if all of its components 
(which are scalar distributions) in any given coordinate atlas converge in the sense of (scalar) distributions. 

If $A \in \Dcal' T_q^p(\MM)$ and $f \in \Ccal^\infty(\MM)$, or else 
if $A \in \TT_q^p(\MM)$ and $f \in \Dcal'(\MM)$, we define the product of $f$ and $A$ as a distribution, 
by setting 
$$
(f \, A)(X_{(1)},\ldots,X_{(p)},\theta^{(1)},\ldots,\theta^{(q)}) 
= f \, A(X_{(1)},\ldots,X_{(p)},\theta^{(1)},\ldots,\theta^{(q)}). 
$$
for all $X_{(i)}\in\XX(\MM)$ and $\theta^{(j)}\in\TT_0^1(\MM)$.

Finally, given a smooth tensor field $A \in \TT_q^p(\MM)$, we can define its extension
$$
A : \Dcal' \X(\MM)\underbrace{\times\XX(\MM)\times\ldots\times\XX(\MM)}_{(p-1)
\text{ times}}\times\underbrace{\YY(\MM)\times\ldots\times\YY(\MM)}_{q\text{ times}}\to\Dcal'(\MM)
$$
by setting
\be\label{distrib-slot}
A(Y,X_{(2)},\ldots,X_{(p)},\theta^{(1)},\ldots,\theta^{(q)})
:= \la Y, A(\cdot ,X_{(2)},\ldots,X_{(p)},\theta^{(1)},\ldots,\theta^{(q)})\ra
\ee
for all $Y\in \Dcal' \X(\MM)$, $X_{(i)}\in\XX(\MM)$, and $\theta^{(j)}\in\YY(\MM)$.
Here, the term 
$$
A(\cdot ,X_{(2)},\ldots,X_{(p)},\theta^{(1)},\ldots, \theta^{(q)})\in\YY(\MM)
$$
 is the $1$-form field defined by
$$
X \in \XX(\MM) \mapsto A(X,X_{(2)},\ldots,X_{(p)},\theta^{(1)},\ldots, \theta^{(q)})\in\Ccal^\infty(\MM).
$$
Extensions corresponding to other slots are defined similarly.


\section{Distributional curvature associated with a connection} 
\label{PG-1}  

\subsection{Distributional connections}
\label{distrib-connections}
 
We begin now our investigation of connections and metrics with limited regularity,
and we consider first a general notion of connection operator $(X,Y)\mapsto \nabla_X Y$
defined in the distribution sense. This is a very general concept for which $\nabla_X Y$ is only a 
distribution vector field.

\begin{definition}\label{covar} An operator 
$\nabla: \XX(\MM) \times \XX(\MM) \to \Dcal' \X(\MM)$
is called a {\rm distributional connection} if it satisfies the linearity and Leibnitz properties
$$
\aligned 
& \nabla_{f X + X'} Y = f \nabla_X Y + \nabla_{X'}Y, 
\\
& \nabla_X(Y+Y')= \nabla_X Y + \nabla_X Y', 
\qquad \nabla_X (fY) = f\nabla_X Y + (Xf)Y, 
\endaligned 
$$
for all $f\in\Ccal^\infty(\MM)$ and $X,X',Y,Y'\in \XX(\MM)$. 
\end{definition}

Observe that such a distributional connection can be extended to act on tensors of general order.
Namely, for functions we trivially define the operator (still denoted by the same symbol) 
$\nabla:\XX(\MM)\times \Ccal^\infty(\MM) \to \Ccal^\infty(\MM)\subset\Dcal'(\MM)$ 
by
$$
\nabla_X f:=X f, \qquad f\in\Ccal^\infty(\MM), \quad X \in \XX(\MM). 
$$
For $1$-form fields we introduce the operator $\nabla:\XX(\MM)\times \YY(\MM) \to \Dcal' \Y(\MM)$ by 
$$
\aligned 
& \prec \nabla_X \theta,Z\succ := X (\la\theta,Z\ra) - \la \theta,\nabla_X Z\ra, 
\\
& Z \in \XX(\MM), \quad X \in\XX(\MM), \quad \theta\in\YY(\MM). 
\endaligned 
$$
Thanks to \eqref{distrib-slot}, the term $\la \theta, \nabla_X Z \ra$ 
above is well-defined as a {\sl scalar} distribution.

Finally, we introduce operator $\nabla:\XX(\MM)\times \TT_q^p(\MM) \to \Dcal' T_q^p(\MM)$, 
defined for $X \in \XX(\MM)$ and $T\in\TT_q^p(\MM)$ by
$$
\aligned
& (\nabla_X T)(Z_{(1)},\ldots,Z_{(p)},\theta^{(1)},\ldots,\theta^{(q)})
\\
& :=X(T(Z_{(1)},\ldots,Z_{(p)},\theta^{(1)},\ldots,\theta^{(q)}))\\
& \quad -\sum_{i=1}^p T(Z_{(1)},\ldots,Z_{(i-1)},\nabla_XZ_{(i)},Z_{(i+1)},\ldots,Z_{(p)},\theta^{(1)},\ldots,\theta^{(q)})\\
& \quad -\sum_{j=1}^q T(Z_{(1)},\ldots,Z_{(p)},\theta^{(1)},\ldots,\theta^{(j-1)},\nabla_X\theta^{(j)},\theta^{(j+1)},\ldots,\theta^{(q)})
\endaligned
$$
for all $Z_{(i)} \in \XX(\MM)$ and $\theta^{(j)}\in\YY(\MM)$. Again, we observe that, 
 thanks to \eqref{distrib-slot},  
the last two terms above are well-defined as distributions in $\Dcal'(\MM)$.

In consequence, a distributional connection enjoys all of the linearity and Leibnitz properties
of smooth connections. 

However, the interest of Definition~\ref{covar} is probably limited by the fact 
that no curvature tensor can be associated with a general distributional connection. 
This is due to the impossibility to multiply two general distributions, operation that is essential 
in defining the curvature.
Indeed, a distributional connection $\nabla$ does not allow one to compute second-order covariant derivatives: 
the vector field $\nabla_XZ$ is only a distribution, {\sl even} if $X,Z$ are smooth vector fields; hence,  
the term $\nabla_X(\nabla_Y Z)$ does not make sense. In coordinates, 
we schematically have 
$$
\Riem = \del \Gamma + \Gamma \star \Gamma, 
$$
where $\Gamma$ stands for Christoffel symbols of the connection $\nabla$. The quadratic products in 
the term $\Gamma \star \Gamma$ can not be defined in the distribution sense.

\subsection{The class of $L_\loc^2$ connections} 
\label{L2-connections}

To identify the connections that do admit a curvature tensor we now restrict attention to 
less singular connections.  
Consider an arbitrary distributional connection $\nabla$ satisfying 
$\nabla_X Y\in E \X(\MM)$ for all $X,Y\in\XX(\MM)$, where $E \X(\MM)$ is some 
subspace of $\Dcal' \X(\MM)$ (to be specified shortly). 
To give a meaning to its curvature tensor we must compute second-order derivatives of vector fields, 
that is, terms like $\nabla_X\nabla_Y Z$.  
Since $\nabla_Y Z$ belongs to $E\X(\MM)$, we first extend the operator $\nabla$ 
to the larger space $\XX(\MM) \times E\X(\MM)$. 
Such an extension, say $\nabla : \XX(\MM) \times E\X(\MM)\to \Dcal'\X(\MM)$, should naturally 
be defined by the formula 
\be
\label{relatio}
\la\nabla_X V,\theta\ra = X\la V,\theta\ra -\la V,\nabla_X \theta\ra
\qquad \text{ in } \Dcal'(\MM),
\ee
for $\theta\in \YY(\MM)$, $V\in E\X(\MM)$, and $X\in\XX(\MM)$. Under our assumptions, 
we solely have $\nabla_X \theta\in E\Y(\MM)$ and $V\in E\X(\MM)$
and, therefore, we see that the term $\la V, \nabla_X \theta\ra$ can be defined as a distribution 
only if $E$ is (a subspace of) $L^2_\loc$. 

This discussion leads us to:  

\begin{definition}
\label{ext-nabla} 
A distributional connection $\nabla$ is called an {\rm $L^2_\loc$ connection} 
if $\nabla_XY\in L^2_\loc\X(\MM)$ for all $X,Y\in\XX(\MM)$. 
The extension of such a connection is the operator $\nabla:\XX(\MM)\times L^2_\loc\X(\MM)\to \Dcal'\X(\MM)$ defined by 
$$
\la\nabla_X Y,\theta\ra  = X(\la Y,\theta\ra) - \la Y,\nabla_X\theta \ra
\quad \text{ in } \Dcal'(\MM),
$$
for $X\in\XX(\MM)$, $Y \in L^2_\loc\X(\MM)$, and $\theta\in \Dcal\Y(\MM)$.    
\end{definition}
 
When the conditions in the definition hold, we write in short $\nabla \in L^2_\loc(\MM)$;  
according to Definition~\ref{ext-nabla} one can then compute 
the covariant derivative of an $L^2$ vector field and, in turn, 
compute the curvature of $\nabla$.

\begin{definition} 
\label{curva}  
The {\rm distributional Riemann curvature} tensor of an $L^2_\loc$ connection $\nabla$  
is the tensor distribution $\Riem:\XX(\MM)\times \XX(\MM)\times \XX(\MM) \to \Dcal'\X(\MM)$ 
defined 
for $\theta\in \YY(\MM)$ and $X,Y,Z \in \XX(\MM)$ by  
$$
\aligned
\la \Riem(X,Y)Z,\theta\ra 
& = X\la\nabla_Y Z, \theta\ra -Y\la\nabla_X Z, \theta\ra
\\
& \quad -\la\nabla_Y Z, \nabla_X\theta\ra+\la\nabla_X Z, \nabla_Y\theta\ra-\la\nabla_{[X,Y]} Z, \theta\ra 
\endaligned
$$
as an equality in $\Dcal'(\MM)$. 
\end{definition}

Provided each term is understood in the distribution sense as explained above, 
we can also write the standard formula 
$$
\Riem(X,Y)Z := \nabla_X \nabla_Y Z - \nabla_Y \nabla_X Z - \nabla_{[X,Y]} Z.  
$$
We then introduce:

\begin{definition} 
\label{ricci}
The {\rm distributional Ricci curvature} tensor associated with an $L^2_\loc$ connection $\nabla$ 
is the tensor distribution $Ric:\XX(\MM)\times \XX(\MM) \to \Dcal'(\MM)$ 
defined for all for all $X,Y \in \XX(\MM)$ by 
$$
\Ric(X,Y):= \la E^{(\alpha)}, \Riem(X,E_{(\alpha)})Y\ra    \qquad \text{ in } \Dcal'(\MM), 
$$ 
where $E_{(\alpha)}$ ($\alpha=1, \ldots, m$) is an arbitrary local frame in the bundle $T\MM$ 
and $E^{(\alpha)}$ ($\alpha=1, \ldots, m$) is the 
corresponding dual frame. 
\end{definition}

The distributional curvature tensors defined above enjoy some important stability properties.  

\begin{theorem}[Stability under strong $L^2_\loc$ convergence]
\label{seq-conn}
Let $\nabla^{(n)}$ ($n=1,2,\ldots$) be a sequence of $L^2_\loc$ connections defined on $\MM$, 
and converging in the $L^2_\loc$ topology to some $L^2_\loc$ connection $\nabla^{(\infty)}$, 
$$
\nabla^{(n)}_XY\to \nabla^{(\infty)}_X Y 
\qquad \text{ strongly in } L^2_\loc
$$
for $X,Y \in \XX(\MM)$. Then, the distributional Riemann and Ricci curvature tensors 
$\Riem^{(n)}$ and $\Ric^{(n)}$ of the connections $\nabla^{(n)}$ converge in the distribution sense 
to the distributional curvature tensors $\Riem^{(\infty)}$ and $\Ric^{(\infty)}$ 
of the limiting connection $\nabla^{(\infty)}$, 
$$
\Riem^{(n)} \to \Riem^{(\infty)}, 
\qquad \Ric^{(n)} \to \Ric^{(\infty)}. 
$$
\end{theorem}

\begin{proof} The desired convergence result follows from the above definitions of distributional curvature 
and on the key identity \eqref{relatio}. 
\end{proof}

\subsection{Jump relations}
\label{jump-connections}

Consider now the case that the connection $\nabla$ suffers a jump discontinuity 
along a smooth hypersurface $\HH\subset\MM$ and is smooth on both sides of it. 
Suppose that the hypersurface splits the manifold into two components, say  
$$
\MM = \MM^- \cup \MM^+, \qquad \MM^- \cap \MM^+ = \HH,  
$$
where $\MM^\pm$ are connected, $\Ccal^\infty$ differentiable manifolds with boundary. 
Denote by $T_x\HH$ and $T_x^\star\HH$ the tangent and cotangent spaces of $\HH$.
Assume that $\MM$ is endowed with a connection $\nabla$ of class 
$L^2_\loc(\MM^\pm) \cap W_\loc^{1,p}(\MM^\pm)$
for some $p \geq 1$
in each component $\MM^\pm$ up the boundary $\HH$, but suffers a jump discontinuity across $\HH$. 
This means that the restriction $(\nabla_X Y)^\pm$ of the vector field $\nabla_X Y$ to $\MM^\pm$ 
is of class $L^2_\loc \cap W_\loc^{1,p}$ in $\MM^\pm$ up to the boundary $\HH$. 
In consequence, the manifolds with boundary $\MM^\pm$ are naturally endowed 
with the connections $\nabla^\pm$ of class $L^2_\loc \cap W_\loc^{1,p}$ defined by 
$$
\nabla^\pm_{X^\pm}Y^\pm:=(\nabla_X Y)^\pm,  \qquad X^\pm,Y^\pm \in \XX(\MM^\pm). 
$$
Here, $X,Y\in\XX(\MM)$ are arbitrary (smooth) vector fields whose restrictions to $\MM^\pm$ 
coincide with $X^\pm,Y^\pm$. One can check that this identity defines $\nabla^\pm$ unambiguously and uniquely. 

Since the operators $\nabla^\pm$ are of class $L^2_\loc \cap W_\loc^{1,p}$, their Riemann and Ricci curvatures 
$\Riem^\pm$ and $\Ric^\pm$ are well-defined in $\MM^\pm$ up to the boundary $\HH$, 
and belong to $L^1_\loc(\MM^\pm)$, and even to $L^p_\loc(\MM^\pm)$ if $p \geq m/2$. 
If $X,Y,Z\in\XX(\MM)$ are smooth vector fields defined over the entire manifold $\MM$, we use the 
short-hand notation 
$$
\aligned
\nabla^\pm_XY   &:=\nabla^\pm_{X^\pm}Y^\pm,\\
\Riem^\pm(X,Y)Z  &:=\Riem^\pm(X^\pm,Y^\pm)Z^\pm,\\
\Ric^\pm(X,Y)    &:=\Ric^\pm(X^\pm,Y^\pm).
\endaligned
$$

Our aim is to rely on Definitions~\ref{curva} and \ref{ricci} 
and compute the {\sl distributional} curvature of the connection $\nabla$, 
which is defined over the entire manifold $\MM$. Observe that 
even if one considers smooth fields $X,Y$, the covariant derivative $\nabla_X Y$ is {\sl not} smooth. 
As observed earlier, in order to compute second-order covariant derivatives of smooth vector fields 
we must compute the covariant derivative of fields having the regularity of $\nabla_X Y$ only. 
The latter suffers a jump discontinuity across the hypersurface $\HH$, 
and we can anticipate that its covariant derivative will contain Dirac mass singularities along $\HH$.
 
Before we can state the corresponding formulas we introduce the following definition.

\begin{definition}
\label{dirac}
The Dirac measure supported by a smooth hypersurface $\HH\subset\MM$ is the $1$-form distribution 
$\Dirac_{\HH} \in \Dcal'\Y(\MM)$ defined by 
$$
\aligned 
& X \in \XX(\MM)\mapsto \la \Dirac_{\HH},X\ra\in \Dcal'(\MM), 
\\
& \lad \la \Dirac_{\HH}, X \ra, \omega \rad = \int_\HH i_X \omega, \qquad \omega \in \Dcal \Lambda^\dM(\MM). 
\endaligned 
$$ 
\end{definition}

\begin{remark} 
Observe that the distribution $\la \Dirac_\HH,X\ra$ depends on $X$ only via its restriction to the 
hypersurface $\HH$; 
therefore, the Dirac measure can be applied on vector fields $X$ that are {\sl only defined on $\HH$}. 
Moreover, if $X$ is a vector field tangent to the hypersurface ($X \in \TT^1_0(\HH)$), 
then the action of the Dirac measure on $X$ is trivial: $\la\Dirac_{\HH},X\ra=0$.
\end{remark}  

It will be convenient to consider a (locally defined, at least) frame of vector fields 
$E_{(\alpha),x}$, $\alpha= 1,\ldots,m$, {\sl adapted to the hypersurface} in the sense that 
$E_{(\alpha),x}$, $\alpha= 1,\ldots,m$ is a basis of $T_x\MM$ for all $x \in \MM$, 
while 
$E_{(i),x}$, $i=1, \ldots, m-1$ is a basis of $T_x \HH$ for all $x \in \HH$. 
Denote also by $E^{(\alpha)}_x$, $\alpha= 1,\ldots,m$ the corresponding dual frame of $1$-form fields, i.e., 
$$
\la E^{(\alpha)}, E_{(\beta)}\ra =  \delta^\alpha_\beta := \begin{cases} 0, & \alpha \neq \beta, 
\\
1, & \alpha = \beta. 
\end{cases}
$$
Greek and latin indices will always describe the range $1, \ldots,m$ and $1,$ $\ldots$, ${m-1}$, respectively.   

One more notation will be useful. We write $[ A ]_\HH$ for the jump of a tensor field $A$ across $\HH$, 
that is, with obvious notation, $[ A ]_\HH := A^+ - A^-$, while 
the ``regular part'' of a distribution tensor field $A$ is expressed as 
$$
A^{reg}:=
\begin{cases}
A^+ & \text{ in } \MM^+,
\\
A^- & \text{ in } \MM^-. 
\end{cases}
$$

\begin{theorem}[Jump relations associated with a singular connection] 
\label{PG-2} 
Let $\HH$ be a smooth hypersurface within a smooth manifold $\MM=\MM^- \cup \MM^+$
separated into two manifolds with boundary.  
Let $\nabla$ be a connection on $\MM$ satisfying  
$\nabla\in L^2_\loc(\MM)$ and $\nabla^\pm\in W^{1,p}_\loc(\MM^\pm)$ for some $p\geq 1$. 
\begin{enumerate}
\item The distributional covariant derivative $\nabla V$ of a vector field $V$ that is smooth in $\MM^\pm$ but discontinuous 
across $\HH$, is given by 
$$
\nabla_X V:=(\nabla_X V)^{reg} +[V]_\HH \, \la \Dirac_{\HH},X\ra, 
\qquad X \in \XX(\MM). 
$$
\item The distributional Riemann curvature of the connection $\nabla$ is the sum 
of a regular part and a Dirac measure supported on $\HH$: for all $X,Y,Z \in \XX(\MM)$ 
$$
\aligned 
\Riem(X,Y)Z = (\Riem(X,Y)Z)^{reg}
& + [\nabla_Y Z ]_\HH \la \Dirac_\HH, X\ra 
\\
& - [\nabla_X Z]_\HH \la\Dirac_\HH, Y\ra. 
\endaligned 
$$ 
\item The distributional Ricci curvature of the connection $\nabla$ is the sum 
of a regular part and a Dirac measure supported on $\HH$: for all $X,Y \in \XX(\MM)$ 
$$
\aligned 
\Ric(X,Y) = (\Ric(X,Y))^{reg} 
& + [\la E^{(\alpha)},\nabla_{E_{(\alpha)}}Y \ra]_\HH \la\Dirac_\HH, X\ra 
\\
& - [\la E^{(m)},\nabla_X Y \ra]_\HH \la\Dirac_\HH, E_{(m)}\ra,
\endaligned 
$$ 
 where $\{E_{(\alpha)}\}$ is a frame adapted to the hypersurface $\HH$. 
\end{enumerate}
The regular parts $(\Riem(X,Y)Z)^{reg}$ and $(\Ric(X,Y)Z)^{reg}$ belong to $L^1_\loc(\MM)$, while the jumps 
$[\nabla_Y Z]_{\HH}$ and $[\nabla_X Z]_{\HH}$ belong to the space $W^{1-1/p,p}_\loc(\HH)$. 
\end{theorem}

\begin{corollary} [Singular parts of curvature tensors]
\label{PG-2b} 
Under the assumptions and notation of Theorem~\ref{PG-2} 
the following hold. 
\begin{enumerate}
\item The singular part in the Riemann curvature vanishes if and only if the connection $\nabla$ is continuous 
across $\HH$ (i.e., its traces from $\MM^\pm$ coincide).  
\item The singular part of the Ricci tensor  
vanishes if and only if, in an adapted frame,
$\la E^{(m)},\nabla_X Y \ra$ and $\la E^{(j)},\nabla_{E_{(j)}}Y \ra$ 
are continuous across $\HH$ for all vector fields $X,Y\in \XX(\MM)$ satisfying $X_x\in T_x\HH$ for all $x\in\HH$. 
\end{enumerate} 
\end{corollary}

\begin{proof}[Proof of Theorem~\ref{PG-2}] 
First of all, the covariant derivative $\nabla_X V$ is, by definition, the vector distribution given by 
\be
\label{cov-1}
\la\nabla_X V,\theta\ra = X\la V,\theta\ra-\la V,\nabla_X \theta\ra, \qquad \theta\in\YY(\MM). 
\ee
The first term in the right-hand side is the scalar distribution defined by
$$
\lad X \la V, \theta \ra, \varphi \rad = - \lad \la V,\theta\ra, \Lcal_X \varphi \rad, 
\qquad \varphi \in \Dcal\Lambda^\dM(\MM). 
$$
Since the restrictions of $\la V,\theta\ra$ to $\MM^\pm$ are smooth, the last relation becomes 
$$
\aligned
\lad X\la V,\theta\ra,\varphi\rad
& =-\int_\MM \la V,\theta\ra \Lcal_X\varphi\\
& =-\int_{\MM^+} \la V^+,\theta\ra \Lcal_X\varphi-\int_{\MM^-} \la V^-,\theta\ra \Lcal_X\varphi
\endaligned
$$
from which, in view of the formula \eqref{formula-div} 
with $\HH$ oriented as the boundary of $\MM^-$, we deduce 
$$
\aligned
\lad X\la V,\theta\ra,\varphi\rad
 & =\int_{\MM^+} X\la V^+,\theta\ra\varphi  +\int_\HH \la V^+,\theta\ra i_X\varphi 
 \\
&  \quad +\int_{\MM^-} X\la V^-,\theta\ra \varphi-\int_\HH \la V^-,\theta\ra i_X\varphi. 
\endaligned
$$
Using this expression in \eqref{cov-1}, we deduce 
$$
\aligned
&\lad\la\nabla_X V,\theta\ra,\varphi\rad 
\\
& = \int_{\MM^+} (X\la V^+,\theta\ra-\la V^+,\nabla^+_X\theta\ra)\varphi
\\
& \quad +\int_{\MM^-} (X\la V^-,\theta\ra-\la V^-,\nabla^-_X\theta\ra)\varphi
    +\int_\HH \la [V]_\HH,\theta\ra i_X\varphi
\\
&= \int_{\MM^+} \la \nabla^+_X V^+,\theta\ra\varphi
    +\int_{\MM^-} \la \nabla^-_X V^-,\theta\ra\varphi
    +\int_\HH \la [V]_\HH,\theta\ra i_X\varphi.
\endaligned
$$
In other words, the covariant derivative of $V$ is the vector distribution
$$
\nabla_X V:=(\nabla_X V)^{reg} +[V]_\HH \, \la \Dirac_\HH, X\ra,
$$
which provides the desired identity for the connection. 

\

Second, by Definition~\ref{curva}, the distributional Riemann curvature is  
$$
\Riem(X,Y)Z = \nabla_X \nabla_Y Z - \nabla_Y \nabla_X Z - \nabla_{[X,Y]} Z 
$$ 
for all  $X,Y,Z \in \XX(\MM)$. Since the vector fields $\nabla_Y Z$ and $\nabla_X Z$ both 
satisfy the regularity assumptions on $V$ of Proposition~\ref{PG-2}, we deduce 
$$
\aligned
& \Riem(X,Y)Z 
\\
& = (\nabla_X \nabla_Y Z)^{reg}
 + [\nabla_Y Z]_\HH \, \la \Dirac_\HH, X\ra ) 
- (\nabla_Y \nabla_X Z)^{reg} +[\nabla_X Z]_\HH \la \Dirac_\HH, Y\ra ) - \nabla_{[X,Y]} Z\\
&=(\Riem(X,Y)Z)^{reg}+ [\nabla_Y Z]_\HH \la \Dirac_\HH, X\ra -[\nabla_X Z]_\HH \la \Dirac_\HH, Y\ra,
\endaligned
$$
where we set
$$
(\Riem(X,Y)Z)^{reg}:=
\begin{cases}
\Riem^+(X,Y)Z    & \text{ in } \MM^+,
\\
\Riem^-(X,Y)Z    & \text{ in } \MM^-.
\end{cases}
$$
This establishes the identity in the theorem. 

Third, from the definition of the Ricci tensor, we obtain 
$$
\Ric(X,Y)=\la E^{(\alpha)}, \Riem(X,E_{(\alpha)})Y\ra,
$$ 
and from the above decomposition of the Riemann curvature tensor, we deduce that 
$$
\aligned 
\Ric(X,Y) =(\Ric(X,Y))^{reg}
& + [\la E^{(\alpha)},\nabla_{E_{(\alpha)}}Y \ra]_\HH \la\Dirac_\HH, X\ra 
\\
& - [\la E^{(\alpha)},\nabla_X Y \ra]_\HH \la\Dirac_\HH, E_{(\alpha)}\ra.
\endaligned 
$$
Now, observe that $\la\Dirac_\HH, E_{(j)}\ra=0$ since the vector fields $E_{(j)}$ are tangent to $\HH$. 
This establishes the desired formula for the Ricci tensor.
\end{proof}

\begin{proof}[Proof of Corollary~\ref{PG-2b}] 
We see immediately that the singular part of the Ricci tensor vanishes if and only if 
$$
\Ric(X,Y)=0, \qquad \Ric(E_{(m)},Y)=0
$$
for all $X,Y\in \XX(\MM)$ such that $X_x\in T_x\HH$ at any point $x\in\HH$. 
But for such vector fields, the singular part of $\Ric(X,Y)$ is 
$$
- [\la E^{(m)},\nabla_X Y \ra]_\HH \la\Dirac_\HH, E_{(m)}\ra, 
$$
while the singular part of $\Ric(E_{(m)},Y)$ is
$$
[\la E^{(\alpha)},\nabla_{E_{(\alpha)}}Y \ra]_\HH \la\Dirac_\HH, E_{(m)}\ra
- [\la E^{(m)},\nabla_{E_{(m)}} Y \ra]_\HH \la\Dirac_\HH, E_{(m)}\ra.
$$
The latter is nothing but $[\la E^{(j)},\nabla_{E_{(j)}}Y \ra]_\HH \la\Dirac_\HH, E_{(m)}\ra$. 
\end{proof}


\section{Distributional curvature associated with a metric}
\label{geometry-metric} 
 
\subsection{Distributional metrics}
\label{distrib-metrics}

To a smooth metric tensor $g$ defined on $\MM$, one associates a unique (Levi-Cevita) connection operator $\nabla$ 
satisfying $\nabla g = 0$ and the zero torsion condition $\Tor=0$, where 
\be
\Tor(X,Y) := \nabla_X Y - \nabla_Y X - [X,Y],  \qquad X,Y \in \XX(\MM).  
\label{G-notorsion}
\ee 
A natural question is whether the notion of Levi-Cevita connection extends to metrics with weak regularity. 
We show now that a distributional connection (in the sense introduced in the previous section) 
can be associated with a distributional metric (in a sense defined now).

\begin{definition}
A {\rm distributional metric} $g$ on $\MM$ is a symmetric and non-degenerate $(0,2)$-tensor 
distribution on $\MM$, that is, $g:\XX(\MM)\times \XX(\MM)\to\Dcal'(\MM)$ satisfying 
$$
\aligned
& g(X,Y)=g(Y,X),
\\
& g(X,Y)=0 \ \text{ for all } Y  \, \Rightarrow \, X=0.
\endaligned
$$
\end{definition}

Further regularity on $g$ will be imposed later on, when necessary. 
It would be natural to define a distributional connection $\nabla$ associated with 
a distributional metric $g$ by requiring the two conditions 
\be
\label{metric-nabla}
\Tor =0, \qquad \nabla g = 0, 
\ee
in a suitably weak sense. The second equation must be handled carefully, 
since (cf.~the discussion in the previous section) non-smooth connections 
do not act on non-smooth tensors such as the metric $g$. 
Therefore, we will not rely directly on the equation $\nabla g = 0$, 
and to circumvent this difficulty we take advantage of the additional structure induced by the metric $g$. 

We make the following two observations valid for {\sl smooth} metrics: 
\begin{enumerate} 
\item The Levi-Cevita connection of a metric $g$ satisfies the Koszul formula
\be
\label{Koszul} 
\aligned
2 \, g(\nabla_XY,Z)=& X(g(Y,Z))+Y(g(X,Z))-Z(g(X,Y))\\
& -g(X,[Y,Z])-g(Y,[X,Z])+g(Z,[X,Y])
\endaligned
\ee 
for all $X,Y,Z \in \XX(\MM)$. This is easily checked from the conditions \eqref{metric-nabla}. 
\item The left-hand side of \eqref{Koszul} takes the equivalent form 
$$
g(\nabla_XY,Z)=\la(\nabla_XY)^\flat,Y\ra,
$$
where $X\in\XX(\MM)\to X^\flat\in\YY(\MM)$ is the ``duality operator'' defined by 
$$
\la X^\flat, Y\ra:=g(X,Y), \qquad Y\in \XX(\MM).
$$  
\end{enumerate}

Thanks to the above observation we see that the connection of a metric can be determined by 
the formula 
\be
\label{koszul}
\aligned
2 \, \la(\nabla_XY)^\flat,Z\ra
= &X(g(Y,Z))+Y(g(X,Z))-Z(g(X,Y))\\
& -g(X,[Y,Z])-g(Y,[X,Z])+g(Z,[X,Y]). 
\endaligned
\ee
Furthermore, the conditions \eqref{metric-nabla} characterizing the Levi-Cevita connection 
read for all $X,Y,Z\in\XX(\MM)$
\be
\label{metric-nabla.bis}
\aligned
& (\nabla_X Y)^\flat-(\nabla_Y X)^\flat -[X,Y]^\flat=0,
\\
& X(g(Y,Z))-\la(\nabla_X Y)^\flat,Z\ra-\la Y,(\nabla_X Z)^\flat\ra=0. 
\endaligned
\ee

We now observe that all these relations feature the term $(\nabla_X Y)^\flat$ and that the right-hand side of 
\eqref{koszul} is well-defined in $\Dcal'(\MM)$, 
not only for smooth metrics but also for distributional metrics. 
This suggests how to associate a connection to a distributional metric. 
Specifically, we extend the definition of the operator 
$$
(X,Y)\in\XX(\MM)\times\XX(\MM)\to (\nabla_X Y)^\flat \in\YY(\MM)
$$
to non smooth metrics in such a way that the equations \eqref{metric-nabla.bis}
are satisfied in the distribution sense.

\begin{definition}  
The {\rm distributional Levi-Cevita connection} of a distributional metric $g$ is the operator 
$\nabla^\flat:(X,Y)\in\XX(\MM)\times\XX(\MM)\mapsto \nabla^\flat_X Y\in\Dcal' \Y(\MM)$, 
defined by 
\be
\label{def-nabla}
\aligned
\la \nabla^\flat_X Y,Z \ra:=\frac12 \Big(
& X(g(Y,Z))+Y(g(X,Z))-Z(g(X,Y))\\
& -g(X,[Y,Z])-g(Y,[X,Z])+g(Z,[X,Y])\Big). 
\endaligned
\ee
\end{definition}

In the following, we will refer to \eqref{def-nabla} as the {\sl dual Koszul formula.} 
It is a simple matter to check that $\nabla^\flat$ does satisfy the relations \eqref{metric-nabla} in a weak sense, namely
\be
\label{metric-nabla.ter}
\aligned
& \nabla^\flat_X Y-\nabla^\flat_Y X - [X,Y]^\flat=0,\\
& X(g(Y,Z))-\la\nabla^\flat_X Y,Z\ra-\la Y,\nabla^\flat_X Z\ra=0,
\endaligned
\ee
for all $X,Y,Z\in\XX(\MM)$. 
Furthermore, when $g$ is smooth, then $\nabla^\flat_X Y=(\nabla_X Y)^\flat$ for all $X,Y\in\XX(\MM)$ and 
we recover the standard requirements that $\Tor = 0$ and $\nabla g =0$.

\begin{theorem}[Stability under convergence in the distribution sense]  
\label{seq-met}
Let $g^{(n)}$ ($n=1$, $2$, $\ldots$) be a sequence of distributional metrics 
converging in the distribution sense to some limiting metric $g^{(\infty)}$. 
Then the distributional connection ${\nabla^\flat}^{(n)}$ associated with $g^{(n)}$
converges in the distribution sense to the connection ${\nabla^\flat}^{(\infty)}$ 
associated with $g^{(\infty)}$, i.e., for all $X,Y,Z\in\XX(\MM)$
$$
\la {\nabla^\flat}^{(n)}_X Y, Z \ra \to \la {\nabla^\flat}^{(\infty)}_X Y, Z \ra 
\quad 
\text{ in }\Dcal'(\MM).
$$
\end{theorem}

\begin{proof} It suffices to use the definition of the convergence in the distribution sense of $g^{(n)}$ 
(see Section~\ref{distributions}) together with the definition of the Levi-Cevita connection, as given by the dual Koszul formula. 
\end{proof}

This result should be compared with Theorem~\ref{seq-conn} where local $L^2$ strong convergence of 
a sequence of connections was assumed to deduce the convergence (in the distribution sense) of their curvatures.  
Here, the convergence of $g^{(n)}$ in the distribution sense 
suffices to imply the convergence of their connections. 
This is due to the fact that, roughly speaking, the connection $\nabla^\flat$ depends ``linearly" upon the metric, 
while the curvature depends ``quadratically'' upon the connection.    

Hence, a distributional metric induces a distributional connection, which generalizes 
the Levi-Civita connection associated with a smooth metric. Without further regularity assumption on the metric, 
the Levi-Cevita connection is a distribution only and, as explained in Section~\ref{distrib-connections} one can not define its 
curvature. This motivates us to now introduce a class of more regular metrics.

\subsection{The class of $H^1$ metric tensors}
\label{H1-metrics}

We now specialize the results in Sections~\ref{L2-connections} and \ref{jump-connections} to the case that 
the connection is determined by a metric. The objective is to identify regularity assumptions on the 
metric guaranteeing that the results of Sections~\ref{L2-connections} and \ref{jump-connections} apply 
to the induced connection and allow us to define the Riemann, Ricci, and scalar curvature.  
We recover, using here a coordinate-free presentation, 
results obtained earlier by Geroch and Traschen \cite{GerochTraschen}. 

We first consider metrics $g$ for which the induced connection is of class $L^2_\loc$, 
which in view of the results in Section~\ref{L2-connections} guarantees that the Riemann curvature tensor 
is well-defined in the distribution sense. To this end, recall that the connection induced by $g$ is the operator
$$
(X,Y)\to \nabla_X Y,
$$
uniquely defined by the Koszul formula
$$
\aligned
2g(\nabla_XY,Z)=& X(g(Y,Z))+Y(g(X,Z))-Z(g(X,Y))\\
& -g(X,[Y,Z])-g(Y,[X,Z])+g(Z,[X,Y]),
\endaligned
$$
which is valid for all $X,Y,Z\in\XX(\MM)$. 

Specifically, if $E_{(\alpha)}$, $\alpha = 1,2,\ldots,m$, is a smooth local frame on $\MM$, then 
$$
\nabla_XY= g^{\alpha\beta} g(\nabla_XY,E_{(\beta)}) E_{(\alpha)},
$$
where $(g^{\alpha\beta}_x)$ is the inverse of the matrix $(g(E_{(\sigma)},E_{(\tau)})_x)$ for all $x\in\MM$. 
In order to have $\nabla_XY\in L^2_\loc$, it suffices for the metric $g$ to satisfy the assumptions 
$$
\aligned
& g(\nabla_XY,Z)\in L^2_\loc, \qquad X,Y,Z\in\XX(\MM),
\\
& g^{\alpha\beta}\in L^\infty_\loc, \qquad \alpha,\beta = 1, \ldots, m. 
\endaligned
$$
The second assumption is satisfied for instance if the metric $g$ is in $L^\infty_\loc$ and is also \emph{
uniformly non-degenerate}, in the sense that 
$$
|\det(g(E_{(\alpha)}, E_{(\beta)}))|\ge C \qquad \text{ in } \MM
$$
for some positive constant $C$, as is checked from the definition of the inverse of a matrix. 

Since $\nabla$ is of class $L^2_\loc$, we can now apply the results of Section~\ref{L2-connections}.

\begin{proposition}
Let $g$ be a uniformly non-degenerate metric of class $H^1_\loc\cap L^\infty_\loc$ over a smooth manifold $\MM$.
 Then 
the following properties holds: 
\begin{enumerate}
\item The connection $\nabla$ defined by the Koszul formula is of class $L^2_\loc$ and satisfies 
$$
\Tor=0, \qquad \nabla g=0.
$$
\item The Riemann and Ricci tensors associated with $g$ are well-defined as distributions, in virtue
of Definitions~\ref{curva} and \ref{ricci}.
\item The scalar curvature of $g$ is well-defined as a distribution on $\MM$ by 
$$
R:=g^{\alpha\beta} \Ric(E_{(\alpha)},E_{(\beta)}).
$$
\end{enumerate}
\end{proposition}

\begin{proof} We only need to prove (3). Since $H^1_\loc\cap L^\infty_\loc$ is an algebra and $g$ is 
uniformly non-degenerate, it follows that $g^{\alpha\beta}\in H^1_\loc\cap L^\infty_\loc(\MM)$. On the other hand,
$$
\aligned
& \Ric(E_{(\alpha)},E_{(\beta)})
  = \la E^{(\sigma)}, \Riem(E_{(\alpha)},E_{(\sigma)})E_{(\beta)}\ra\\
& 
   = E_{(\alpha)}(\la E^{(\sigma)}, \nabla_{E_{(\sigma)}}E_{(\beta)}\ra )
  - E_{(\sigma)}(\la E^{(\sigma)}, \nabla_{E_{(\alpha)}}E_{(\beta)}\ra )\\
& \quad 
- \la \nabla_{E_{(\alpha)}}E^{(\sigma)}, \nabla_{E_{(\sigma)}}E_{(\beta)}\ra
+ \la \nabla_{E_{(\sigma)}}E^{(\sigma)}, \nabla_{E_{(\alpha)}}E_{(\beta)}\ra
- \la E^{(\sigma)}, \nabla_{[E_{(\alpha)},E_{(\sigma)}]}E_{(\beta)}\ra.
\endaligned
$$
Since the last three terms belong to $L^1_\loc(\MM)$, we only need to define the product 
of $g^{\alpha\beta}$ and the distributions 
$$
E_{(\alpha)}(\la E^{(\sigma)}, \nabla_{E_{(\sigma)}}E_{(\beta)}\ra ), 
\qquad 
E_{(\sigma)}(\la E^{(\sigma)}, \nabla_{E_{(\alpha)}}E_{(\beta)}\ra ).
$$
This is done by letting 
$$
g^{\alpha\beta}E_{(\alpha)}(\la E^{(\sigma)}, \nabla_{E_{(\sigma)}}E_{(\beta)}\ra )
:=E_{(\alpha)}(g^{\alpha\beta} \la E^{(\sigma)}, \nabla_{E_{(\sigma)}}E_{(\beta)}\ra )
- (E_{(\alpha)}g^{\alpha\beta})\la E^{(\sigma)}, \nabla_{E_{(\sigma)}}E_{(\beta)}\ra,
$$
$$
g^{\alpha\beta}E_{(\sigma)}(\la E^{(\sigma)}, \nabla_{E_{(\alpha)}}E_{(\beta)}\ra )
:=E_{(\sigma)}(g^{\alpha\beta} \la E^{(\sigma)}, \nabla_{E_{(\alpha)}}E_{(\beta)}\ra )
- (E_{(\sigma)}g^{\alpha\beta})\la E^{(\sigma)}, \nabla_{E_{(\alpha)}}E_{(\beta)}\ra,
$$
which are clearly distributions. 
\end{proof}

\begin{remark} Alternatively, the distributional Riemann curvature of a metric $g$ of class $H^1_\loc\cap L^\infty_\loc$ 
can also be defined as the distribution $\Riem:\XX(\MM)\times \XX(\MM)\times \XX(\MM)\times \XX(\MM)\to \Dcal'(\MM)$ given by 
$$
\aligned
\Riem(W,Z,X,Y):=
& X(g(W,\nabla_Y Z))-Y(g(W,\nabla_X Z))\\
&-g(\nabla_X W, \nabla_Y Z)+g(\nabla_Y W, \nabla_X Z)-g(W,\nabla_{[X,Y]}Z)
\endaligned
$$
for all $X,Y,Z,W\in\XX(\MM)$.  The vector field $\Riem(X,Y)Z$ defined in Definition~\ref{curva} and, for every $W$,  
the function $\Riem(W,Z,X,Y)$ are then related as follows: 
$$
\Riem(W,Z,X,Y) = g(W,\Riem(X,Y)Z), \qquad X,Y,Z,W \in \XX(\MM).
$$
\end{remark}

We can also prove: 
 
\begin{theorem}[Stability under strong $H^1_\loc$ convergence]  
\label{seq-met2bis}
Let $g^{(n)}$ ($n=1, 2, \ldots$) be a sequence of $H^1_\loc$ metric tensors converging locally 
in the strong $H^1_\loc$ topology to some limiting metric $g^{(\infty)}$. 
Assume that the inverse metrics $g^{-1}_{(n)}$ converge locally in the strong $L^\infty_\loc$ topology 
to $g^{-1}_{(\infty)}$. 
\begin{enumerate}
\item The Levi-Cevita connections ${\nabla}^{(n)}$ associated with $g^{(n)}$ are of class $L^2_\loc$ 
and converge in the strong $L^2_\loc$ topology to the connection ${\nabla}^{(\infty)}$ of 
the limit metric $g^{(\infty)}$, that is, for all $X,Y \in \XX(\MM)$,  
$$
{\nabla}^{(n)}_XY\to {\nabla}^{(\infty)}_X Y \qquad \text{ strongly in } L^2_\loc.
$$
\item The distributional Riemann, Ricci, and scalar curvature tensors $\Riem^{(n)}$, $\Ric^{(n)}$, 
$R^{(n)}$ of 
the connections $\nabla^{(n)}$ converge in the distribution sense 
to the limiting curvature tensors $\Riem^{(\infty)}$, $\Ric^{(\infty)}$, $R^{(\infty)}$ of $\nabla^{(\infty)}$. 
\end{enumerate}
\end{theorem}

\begin{proof} The first part of the theorem follows from the definition of the Levi-Cevita connection, 
as given by the Koszul formula. The second part is a direct consequence of Theorem~\ref{seq-conn}. 
\end{proof}

\subsection{Jump relations}
\label{jump-connections2}

We now consider the special case that $g$ is continuous but its derivatives have jump discontinuities across a hypersurface $\HH$. 
Specifically, with the notation in Section~\ref{jump-connections} we assume that $g$ is continuous on $\MM$ and 
has Sobolev regularity in each component $\MM^\pm$ up to the boundary $\HH$. This implies that the distributional Levi-Civita connection 
$\nabla$ is well-defined up to the boundary $\HH$, and $\nabla$ satisfies the assumptions of Section~\ref{jump-connections}. 
Hence, relying on Theorem~\ref{PG-2} we arrive at:

\begin{theorem} [Jump relations associated with a singular metric] 
\label{PG-2-bis} 
Let $\HH \subset \MM$ be a smooth hypersurface separating two sub-manifolds with boundary $\HH$ in a smooth manifold 
$\MM=\MM^- \cup \MM^+$. Let $g$ be a distributional metric that is continuous over $\MM$ such
that $g|_{\MM^\pm}\in W^{2,p}_\loc(\MM^\pm)$ for some $p\geq m/2$. 
Then, the following properties hold: 
\begin{enumerate}
\item The singular part of the Riemann curvature vanishes if and only if $g$ is of class $W^{1,p}_\loc(\MM)$.
\item The singular part in the Ricci tensor vanishes if and only if $\la E^{(m)},\nabla_X Y \ra$ and $\la E^{(j)},\nabla_{E_{(j)}}Y \ra$ 
both are continuous across $\HH$ (i.e., their traces from $\MM^\pm$ coincide)
for all vector fields $X,Y\in \XX(\MM)$ that satisfy $X_x\in T_x\HH$, $x\in\HH$.  
\item The scalar curvature of $g$ is is the sum of a regular part and a Dirac measure supported on $\HH$, i.e., 
$$
R := R^{reg}+[\la g^{m\beta}E^{(j)}-g^{j\beta} E^{(m)},\nabla_{E_{(j)}}E_{(\beta)}\ra]_\HH \la\Dirac_\HH, E_{(m)}\ra. 
$$
\end{enumerate} 
Moreover, the regular parts
$$
(\Riem(X,Y)Z)^{reg}, \qquad   (\Ric(X,Y)Z)^{reg}, \qquad R^{reg} 
$$
belong to $L^p_\loc(\MM)$, while the jump parts 
$$
[\la g^{m\beta} E^{(j)}- g^{j\beta}E^{(m)} ,\nabla_{E_{(j)}}E_{(\beta)}\ra]_\HH
$$  
belong to $W^{1-1/p,p}_\loc(\HH)$.
\end{theorem}

\begin{proof} First of all, we have seen in Theorem~\ref{PG-2} that the singular part of 
the Riemann curvature vanishes if and only if $\nabla$ is continuous across $\HH$. On the other hand, 
the dual Koszul formula shows that $\nabla$ is continuous if and only if 
the first order derivatives of $g$ are continuous across $\HH$.   
(To see this, it suffices to make particular choices of the fields $X,Y,Z$.)  

The second result is a direct consequence of Theorem~\ref{PG-2}. 

Third, by using the expression of the Ricci tensor given by Theorem~\ref{PG-2} 
together with the definition of the scalar curvature, we  obtain 
$$
\aligned 
R = R^{reg}  
& + g^{\tau\beta}[\la E^{(\alpha)},\nabla_{E_{(\alpha)}}E_{(\beta)} \ra]_\HH \la\Dirac_\HH, E_{(\tau)}\ra 
\\
& - g^{\tau\beta}[\la E^{(m)},\nabla_{E_{(\tau)}} E_{(\beta)} \ra]_\HH \la\Dirac_\HH, E_{(m)}\ra. 
\endaligned 
$$
But, $\la\Dirac_\HH, E_{(j)}\ra=0$ for all $j=1,2,\ldots,\dH$, so 
$$
R = R^{reg} + [\la g^{m\beta} E^{(\alpha)}- g^{\alpha\beta}E^{(m)} ,\nabla_{E_{(\alpha)}}E_{(\beta)}\ra]_\HH \la\Dirac_\HH, E_{(m)}\ra.  
$$
The desired formula for the scalar curvature follows. 
The regularity of the curvature tensors follows from Theorem~\ref{PG-2}, 
since the connection $\nabla$ associated with such a metric is of class $W^{1,1}_\loc(\MM^\pm)$
if $g \in H^1_\loc \cap W^{1,p}_\loc$.  
\end{proof}

\subsection{Vacuum spacetimes}

In the vacuum, the Einstein equations reduce to the Ricci-flat condition
$$
\Ric = 0
$$ 
This condition restricts the type of jumps that are allowed across a hypersurface. Consider, for instance, 
a smooth hypersurface $\HH \subset \MM$ separating two submanifolds with boundary $\HH$ in a smooth manifold $\MM=\MM^- \cup \MM^+$. 
If the metric $g$ is continuous over $\MM$, smooth in $\MM^\pm$, but has discontinuous derivatives 
across the hypersurface $\HH$, then the Ricci-flat condition implies that, 
for the Levi-Civita connection induced by $g$ and for all $Y\in\XX(\MM)$,  
$$
\la E^{(m)},\nabla_{E_{(j)}} Y\ra \text{ and } \la E^{(j)},\nabla_{E_{(j)}} Y\ra \text{ are continuous across } \HH.
$$
In the particular case that the hypersurface $\HH$ is nowhere null, this condition implies that the metric $g$ 
must be at least of class $\Ccal^1$. On the other hand, if $\HH$ is a null hypersurface, 
then the derivatives of the metric $g$ need not be continuous across $\HH$.


 \section{Geometry induced on a hypersurface by a connection}
\label{HY-0}

\subsection{Preliminaries} 
\label{ss51}

We now turn our discussion to the geometry of smooth and oriented hypersurfaces 
within a connected, oriented, $\Ccal^\infty$-diffe\-ren\-tiable $m$-manifold $\MM$. 
The arguments presented below also apply to the case that $\HH$ is the boundary of a manifold with boundary. 
We begin our discussion by fixing 
a differentiable manifold endowed with a connection $\nabla$ with limited regularity and, in Section~\ref{CU-0}
below, 
we specialize the results to the case that the connection is determined by a Lorentzian metric $g$.  

Several difficulties arise with Lorentzian metrics, which contrasts with what happens with Riemannian metrics. 
First, a Lorentzian metric on a manifold does not always induce a non-degenerate metric on a hypersurface 
via the usual pull-back.
Second, a connection on a manifold does not induce directly any useful geometry on a hypersurface, 
since, in general, $\nabla_X Y\not\in T\HH$ even if $X,Y\in T\HH$. 
Moreover, no canonical projection of $T_x\MM$ on $T_x\HH$ is available in general. 

Recall that prescribing a connection is equivalent to prescribing a parallel transport of tangent vectors. 
Therefore, the choice of a parallel transport on a manifold clearly does not induce a rule for parallel transporting vectors 
tangent to a submanifold.  
Our objective will be to circumvent these difficulties by using a suitable concept of 
projection from $T_x\MM$ into $T_x\HH$. 

Throughout our investigation, it will be important to distinguish between several operations of ``restriction.'' 
On one hand, given a (scalar, vector field, general tensor) field 
defined at every point of $\MM$ we can obviously restrict it to points on the hypersurface $\HH$. On the other hand, 
given a tensor field defined on the whole tangent and cotangent spaces $T_x\MM, T_x^\star\MM$, $x\in\HH$, 
we can restrict it to the spaces $T_x\HH, T_x^\star\HH$. Both restrictions arise in our discussion. 

All formulas will be expressed in a coordinate-free form and, when necessary, all calculations will be performed 
in a moving frame of the tangent space and in its dual frame of the cotangent space. The components in these moving frames of 
a vector field $X$ and a $1$-form field $\omega$ will be denoted by $X^\alpha$ and $\omega_\alpha$, respectively. 
Indices are raised and lowered by using the metric tensor $g_{\alpha\beta}$ and its inverse, denoted by $g^{\alpha\beta}$. 
In particular, the functions $X_\beta=g_{\beta\alpha}X^\alpha$ defined in this fashion are the components of a $1$-form, 
denoted $X^\flat$, and the functions $\omega^\beta=g^{\beta\alpha}\omega_\alpha$ are the components of a vector field, denoted $\omega^\sharp$.


\subsection{Rigging field and projection operators}
\label{ss52} 

We consider the tangent and cotangent spaces $T_x\MM, T_x^\star\MM$ at $x\in\MM$ and 
the tangent and cotangent spaces $T_x\HH$
and $T_x^\star\HH$ at $x\in\HH$. 
Recall first that, for every $x \in \HH$, the tangent space of $\HH$ at $x$
can be viewed as a subspace of the tangent space of $\MM$, 
\be
\label{tangent}
T_x\HH  \subset T_x\MM,
\ee
since $T_x \HH$ consists of (equivalence classes of) paths restricted to lie in $\HH \subset \MM$. 
By contrast, no similar canonical inclusion is available for the cotangent space. 
Indeed, a form $\alpha \in T_x^\star\HH$
is defined solely on $T_x\HH$ and cannot be canonically extended to the whole of $T_x\MM$.

In view of \eqref{tangent}, 
the most fundamental object one can associate to $\HH$ is a {\sl normal form} $n \in \YY(\MM)$, 
which is a $1$-covariant tensor field $x \mapsto n_x$ defined on the whole of $\MM$ 
and whose restriction on $\HH$ is uniquely characterized (up to a scalar multiplicative factor) by the conditions  
($x \in \HH$) 
\be
\la n_x, X \ra \, 
\begin{cases} 
= 0,     &  X \in T_x\HH, 
\\
\neq 0,    &  X \notin T_x \HH. 
\end{cases}  
\label{HY.normalform}
\ee

To determine a canonical decomposition of the cotangent space, the normal form must be supplemented by a rule to 
identify $T^\star\HH$ to a subset of $T^\star\MM$. This motivates the following definition which was discussed 
in \cite{BarrabesIsrael,ClarkeDray,Taub} and, more recently, in Mars and Senovilla \cite{MarsSenovilla}.  

\begin{definition} 
\label{HY-rigging}
A {\rm rigging vector field} along $\HH$ is a vector field $x \in \HH \mapsto \ell_x \in T_x\MM$ satisfying 
$$
\ell_x \notin T_x \HH, \qquad \la n_x, \ell_x\ra =1. 
$$
\end{definition}

The prescription of a rigging $\ell$ allow us to decompose the tangent space to $\MM$ at 
a point of the hypersurface, as follows 
\be 
T_x \MM = \Vect (\ell_x) \oplus T_x\HH, \qquad x \in \HH,  
\label{HY.directtangent}
\ee 
where $\Vect (\ell_x)$ is the vector space generated by $\ell_x$. 
Hence, given any point $x \in \HH$, to any vector $X_x \in T_x\MM$ one can associate 
its {\sl rigging projection} (or projection in the direction of the rigging), 
$\ot X_x \in T_x\HH$, so that 
$$
X = \la n, X \ra \, \ell +  \ot X. 
$$

Analogously, we can decompose the cotangent space to $\MM$ at a point of the hypersurface, 
as follows 
\be
T_x^\star\MM = \Vect (n_x) \oplus T_x^\star\HH,  \qquad x \in \HH, 
\label{HY.directcotangent}
\ee 
where we are now able to identify the fibers of the cotangent bundle $T^\star\HH$ 
with fibers of the (restriction to $\HH$ of the) cotangent bundle $T^\star\MM$,  
\be
\label{dualspace}
T_x^\star\HH = \big\{ \theta \in T_x^\star\MM \, / \,  \la \theta, \ell_x \ra = 0 \big\}, \qquad x \in \HH. 
\ee
Hence, given any point $x \in \HH$ and any $1$-form $\theta$,
we determine the {\sl normal projection} (or projection 
in the direction of the normal form), $\ut\theta \in T^\star \HH$, so that 
$$
\theta = \la \theta, \ell \ra \, n +  \ut \theta.
$$ 

It must be observed that both the projection operators above involve the normal form and 
do depend on the choice of the rigging vector. However, the projected form $\undertilde \theta$ 
is independent of this choice, since 
$$
\la \undertilde \theta, X \ra = \la \theta, X \ra, \qquad X \in T\HH. 
$$ 


\subsection{Projections expressed in a local frame}
\label{ss53}

We now introduce bases of the tangent and cotangent spaces that are adapted to the projection operators. 
At each $x \in \HH$, we supplement the vector $\ell_x$ with a basis $E_{(a),x}$ ($a=1,2,\ldots,m-1$) of the tangent space $T_x\HH$. 
The form $n_x$ is then naturally supplemented with the corresponding dual basis 
$\Edual^{(a)}_x$ ($a=1,2,\ldots,m-1$), so that the frames are characterized by the orthogonality conditions 
\be
\aligned 
& \ell^\alpha n_\alpha = 1, \qquad \ell^\alpha \Edual^{(a)}_\alpha =0,  
\\
& E_{(a)}^\alpha  n_\alpha=0,  \qquad E_{(a)}^\alpha \Edual^{(b)}_\alpha = \delta_a^b, 
\endaligned 
\label{HY.basis}
\ee
where $\delta_a^b$ is the standard Kronecker symbol. 

To compute the projections, we introduce the tensor
$$
{P_\alpha}^\beta := {\delta_\alpha}^\beta - n_\alpha \, \ell^\beta, 
$$
so that the components of the projection of a vector $X$ and a form $\theta$ read  
\be
\tX^\alpha := {P_\beta}^\alpha \, X^\beta, \qquad \uttheta_\alpha = {P_\alpha}^\beta \theta_\beta. 
\label{HY-projection}
\ee 
Clearly, the projections of the vector $\ell$ and the form $n$ vanish identically. 
The projected vectors and forms lie in $T_x\HH$ and $T_x^\star\HH$, respectively and, therefore, 
can be alternatively expressed in the corresponding bases $E_{(a)}$ and $\Edual^{(a)}$, respectively. 
So, we will also write $\ot X=\ot X^a E_{(a)}$ and $\ut \theta=\ut \theta_a \Edual^{(a)}$, where  
the components of the projected vectors and forms are determined by 
\be 
\tX^a := X^\alpha \Edual_\alpha^{(a)}, \qquad \uttheta_a := \theta_\alpha E_{(a)}^\alpha. 
\label{HY-projection1} 
\ee
Recall that Greek and Latin indices describe $1,\ldots,m$ and $1$, $\ldots$, ${m-1}$, respectively.


\subsection{Connections on $\HH$ induced by projecting $\nabla$}
\label{ss54}

The projection operators defined in Section \ref{ss52} lead us naturally to introduce a connection on the hypersurface 
which depends on the choice of the rigging vector $\ell$. Denoted by $\ut \nabla$, this connection is defined by simply 
projecting the original connection $\nabla$, i.e., 
\be
\utnabla_X Y := \widetilde{\nabla_X Y}, \qquad X,Y \in \TT^1_0(\HH). 
\label{HY-derivative}
\ee
We will refer to it as  
the {\sl projected connection.} 
The tilde symbol is placed below, and this notation will be justified in Subsection~\ref{ss65}, that is, 
when $\nabla$ is the Levi-Cevita connection associated with a metric $g$ and $\ell$ is the normal vector field to a 
non-null hypersurface $\HH$, then $\utnabla$ is the Levi-Cevita connection associated with the pull-back $\utg$ of the metric $g$.

\begin{proposition}[Properties of the projected connection] 
\label{HY-connection} 
1. $\utnabla$ is a connection operator, i.e., satisfies the linearity and Leibnitz properties
$$
\aligned 
& \utnabla_{\lam X + \lam' X'} Y = \lam \utnabla_X Y + \lam' \utnabla_{X'}Y, 
\\
& \utnabla_X (Y+Y')=\utnabla_X Y +\utnabla_X Y', 
\quad 
\utnabla_X (\lam\, Y) = \lam \, \utnabla_X Y + X(\lam) \, Y, 
\endaligned 
$$
for all vector fields $X,X',Y,Y' \in \TT^1_0(\HH)$ and all smooth functions $\lam, \lam' : \HH \to \RR$. 

2. If $\nabla$ is of class $W^{k,p}_\loc(\MM)$ for some $k,p\geq 1$, then $\utnabla$ is of class $W^{k-1/p,p}_\loc(\HH)$.

3. The torsion of $\utnabla$ satisfies 
\be
\undertilde\Tor(X,Y)=\widetilde{\Tor(X,Y)}, \qquad X,Y \in \TT^1_0(\HH).
 \label{G-torsion}
\ee
As a consequence, if the connection $\nabla$ has zero torsion, then the projected connection $\utnabla$ also has zero torsion.  
\end{proposition}

\begin{proof} Since $\utnabla_X Y$ is the projection of $\nabla_X Y$ on $T\HH$, in order to prove the second assertion of 
the theorem it suffices to prove the following property:
if a vector field $Z$ belongs to $W^{k,p}_\loc T^1_0(\MM)$, then $\ut Z$ belongs to $W^{k-1/p,p}_\loc T^1_0(\HH)$. 
However, this porperty follows immediately from \eqref{HY-projection} and the usual trace properties of Sobolev functions. 

Since the Lie bracket on $\HH$ of two vector fields $X,Y\in \TT^1_0(\HH)$ coincides with the Lie bracket on $\MM$ of the same 
vector fields, the vector field $(\utnabla_X Y - \utnabla_Y X - [X,Y])\in T\HH$ 
is the projection on $T\HH$ of the vector field $(\nabla_X Y - \nabla_Y X - [X,Y])\in T\MM$. This means that 
$\undertilde\Tor(X,Y)=\widetilde{\Tor(X,Y)}$. In particular, this relation shows that if the connection $\nabla$ has zero torsion, 
then the projected connection $\utnabla$ on $T\HH$ also has zero torsion. 

\end{proof}


\subsection{Second fundamental form}
\label{ss55} 

Let us now introduce the {\sl second fundamental form}
(also called the ``shape tensor'') of the hypersurface as the 2-covariant tensor field $K$ defined by
$$
K(X,Y) := \la \nabla_X n, Y\ra, \quad  X,Y\in \TT^1_0(\HH). 
$$
Since $\la \nabla_X n, Y \ra = \nabla_X(\la n, Y\ra)-\la n, \nabla_X Y\ra$, we also have 
$$
K(X,Y)=-\la n, \nabla_X Y\ra , \quad  X,Y\in \TT^1_0(\HH).
$$
The tensor $K$ is the pull-back of the 2-covariant tensor field $\nabla n$ of $\MM$, where $n\in \TT^0_1(\MM)$ is any 
extension of the normal form $n$ outside the hypersurface (such an extension always exists and the definition of $K$ 
does not depend on the choice of the extension). 

\begin{lemma}
\label{HY-frame} 
1. On the hypersurface $\HH$, the connection $\nabla$ can be expressed in terms of 
 $\utnabla$ and $K$, as follows~:  
\be
\nabla_X Y = \utnabla_X Y - K(X,Y) \, \ell, \quad  X,Y \in \TT^1_0(\HH).
\label{HY-derivative1}
\ee
2. If $\nabla$ is of class $W^{k,p}_\loc(\MM)$ for some $k,p\geq 1$, 
then $K$ is of class $W^{k-1/p,p}_\loc(\HH)$. \\
3. The second fundamental form satisfies the relation 
$$
K(Y,X)-K(X,Y)=\la n, \Tor(X,Y) \ra , \quad X,Y \in \TT^1_0(\HH). 
$$ 
In particular, if $\nabla$ has zero torsion, then the second fundamental form is a symmetric $2$-covariant tensor.
\end{lemma} 

\begin{proof} 
The decomposition \eqref{HY.directtangent} of the tangent space $T\MM$ shows that 
$$
\nabla_X Y = \utnabla_X Y + \la n,\nabla_X Y\ra \, \ell. 
$$
Combining this relation with the definition of the second fundamental form gives the formula \eqref{HY-derivative1}.

The regularity of $K(X,Y)$, $X,Y\in \TT^1_0(\HH)$, follows from Proposition~\ref{HY-connection}. 

Since $[X,Y]\in \TT^1_0(\HH)$ for all $X,Y\in \TT^1_0(\HH)$, we have 
$$
\aligned
K(Y,X)-K(X,Y) & =\la n,\nabla_X Y-\nabla_Y X\ra \\
              & =\la n,[X,Y]+\Tor(X,Y)\ra
               =\la n,\Tor(X,Y)\ra.
\endaligned
$$
\end{proof}

Beside the second fundamental form, we also introduce the Christoffel symbols
associated with $\nabla$.  
Decomposing the vector fields $\nabla_{E_{(a)}}E_{(b)}\in T\MM$ 
and $\nabla_{E_{(a)}}\ell \in T\MM$ on the basis $\{E_{(c)},\ell\}$, we set 
\be
\aligned
&  \nabla_{E_{(a)}} E_{(b)} = \Gamma_{ab}^c \, E_{(c)} - K_{ab} \, \ell,
\\
&  \nabla_{E_{(a)}} \ell = L_a^c \, E_{(c)} + M_a \, \ell,
\endaligned
\label{HY-movframe}
\ee
where the coefficients $\Gamma_{ab}^c$, $K_{ab}$, $L_a^c$, and $M_a$ are functions defined on $\HH$. 
In turn, these equations give the following decomposition of $1$-form fields:
\be
\aligned
&  \nabla_{E_{(a)}} \Edual^{(b)} = - \Gamma_{ac}^b \, \Edual^{(c)} - L_a^b \, n,
\\
&  \nabla_{E_{(a)}} n = K_{ac} \, \Edual^{(c)} - M_a \, n.
\endaligned
\label{HY-movdualframe}
\ee

From the equation \eqref{HY-derivative1} and the fact that the connection $\utnabla$ has zero torsion, we immediately
deduce that~: 

\begin{lemma}
\label{HY-componentwise}
The projected connection $\utnabla$, the second fundamental form $K$, and the Lie bracket are related to the above 
coefficients via the following formulas:
$$
\aligned 
& \utnabla_{E_{(a)}}E_{(b)}=\Gamma_{ab}^c E_{(c)}, \quad
K(E_{(a)},E_{(b)})=K_{ab}, 
\\
& [E_{(a)},E_{(b)}]=\left(\Gamma_{ab}^c-\Gamma_{ba}^c\right)E_{(c)}.
\endaligned 
$$
\end{lemma}

\subsection{Gauss and Codazzi equations} 
\label{ss56}

We now turn to the discussion of the properties of the Riemann curvature tensors $\Riem$ and $\utRiem$, which are 
naturally associated with the connections $\nabla$ and $\utnabla$, respectively. We assume for simplicity that $\nabla$ has zero torsion. 
 
Before we can relate the curvature tensor $\Riem$ on $\MM$ with the curvature tensor $\utRiem$ on the hypersurface $\HH$ 
we first recall the definition of the Riemann tensor: 
\be
\Riem(X,Y)Z = \nabla_X \nabla_Y Z - \nabla_Y \nabla_X Z - \nabla_{[X,Y]} Z, 
\label{HY-Riemann}
\ee
where $[X,Y]$ is the Lie bracket. Choosing now $X,Y,Z\in \TT^1_0(\HH)$ and using of \eqref{HY-derivative1},
 we compute 
$$
\aligned
\nabla_X \nabla_Y Z 
& = \nabla_X \left(\utnabla_Y Z -K(Y,Z)\,\ell\right) = \utnabla_X \utnabla_Y Z-K(X,\utnabla_Y Z)\,\ell - \nabla_X(K(Y,Z)\,\ell)
\\
& = \utnabla_X \utnabla_Y Z -K(Y,Z)\nabla_X\ell - \Big( K(X,\utnabla_Y Z) + X(K(Y,Z)) \Big) \ell, 
\endaligned
$$
and 
$$
\nabla_{[X,Y]} Z = \utnabla_{[X,Y]} Z - K([X,Y], Z) \, \ell. 
$$ 
We deduce that
$$
\aligned 
& \Riem(X,Y) Z 
\\
& = \utRiem(X,Y) Z - K(Y,Z)\nabla_X\ell +K(X,Z)\nabla_Y\ell\\
& \quad - \Big( K(X,\utnabla_Y Z) - K(Y,\utnabla_X Z)+ X(K(Y,Z)) - Y(K(X,Z)) - K([X,Y],Z) \Big) \ell.
\endaligned 
$$
Since
$$
\aligned 
& X(K(Y,Z))- K(Y,\utnabla_X Z)=(\utnabla_X K)(Y,Z)+K(\utnabla_X Y,Z),
\\
& Y(K(X,Z))- K(X,\utnabla_Y Z)=(\utnabla_Y K)(X,Z)+K(\utnabla_Y X,Z),
\endaligned 
$$
and since the torsion of the connection $\utnabla$ vanishes, we finally obtain 
\be
\aligned 
\Riem(X,Y) Z  = 
& \utRiem(X,Y) Z - K(Y,Z)\nabla_X\ell +K(X,Z)\nabla_Y\ell
\\
& - \Big( (\utnabla_X K)(Y,Z) - (\utnabla_Y K)(X,Z) \Big) \ell.
\endaligned 
\label{HY-curvatureidentity}
\ee

Another useful relation is derived from \eqref{HY-Riemann}
by taking $X,Y\in \TT^1_0(\HH)$ and $Z=\ell$, that is, 
\be
\aligned
& \Riem(X,Y)\ell 
\\
& = \utnabla_X\left(\ot{\nabla_Y \ell}\right) - \utnabla_Y\left(\ot{\nabla_X \ell}\right) 
    -\ot{\nabla_{[X,Y]}\ell} - \la n,\nabla_X\ell\ra\ot{\nabla_Y\ell} + \la n,\nabla_Y\ell\ra\ot{\nabla_X\ell}
\\
& \quad + \Big( X(\la n,\nabla_Y\ell\ra) - Y(\la n,\nabla_X\ell\ra) - \la n,\nabla_{[X,Y]}\ell\ra 
- K(X,\ot{\nabla_Y\ell}) + K(Y,\ot{\nabla_X\ell}) \Big) \ell.
\endaligned
\label{HY-codazzi}
\ee

We are now in a position to contract the general identity \eqref{HY-curvatureidentity} 
with an arbitrary $1$-form field $\theta$ among $n, \Edual^{(a)}$, while the vectors fields $X,Y,Z$ are 
chosen arbitrarily among $E_{(a)}$. Likewise, the general identity 
\eqref{HY-codazzi} can be contracted with an arbitrary $1$-form field $\theta$ among 
$n, \Edual^{(a)}$, while the vectors fields $X,Y$ can be chosen arbitrarily among $E_{(a)}$. 
As usual, the components of the Riemann curvature tensor $\utRiem$ are defined by 
$$
\utRiem(E_{(a)},E_{(b)})E_{(c)} = \utR^d_{cab} \, E_{(d)}. 
$$ 
The regularity assumption in the next theorem is such that the equalities \eqref{HY-curvatureidentity} 
and \eqref{HY-codazzi} have traces on the hypersurface $\HH$. 

\begin{theorem} 
\label{55} 
Assume that $\nabla$ is of class $W^{k,p}_\loc(\MM)$, with $k\geq 2$ and $p>m/k$. 
Then the Riemann curvature tensor $\Riem$ is of class $W^{k-1,p}_\loc(\MM)$, the connection $\ut \nabla$ 
and the functions $\Gamma_{ab}^c$, $K_{ab}$, $L_a^c$, $M_a$ are of class $W^{k-1/p,p}_\loc(\HH)$, and 
together they satisfy the following relations:  

1. Choosing $X=E_{(a)}, Y=E_{(b)}, Z=E_{(c)}$ and $\theta=\Edual^{(d)}$ in \eqref{HY-curvatureidentity}
one obtains the {\rm Gauss equations} 
\be
\Edual_\delta^{(d)} R^\delta_{\gamma\alpha\beta} E_{(c)}^\gamma E_{(a)}^\alpha E_{(b)}^\beta  
= \utR^d_{cab} + K_{ac} L_b^d - K_{bc} L_a^d. 
\label{HY-Gauss} 
\ee

2. Choosing $X=E_{(a)}, Y=E_{(b)}, Z=E_{(c)}$ and $\theta=n$ in \eqref{HY-curvatureidentity} one obtains the {\rm Codazzi-$1$ equations} 
\be
n_\delta R^\delta_{\gamma\alpha\beta} E_{(c)}^\gamma E_{(a)}^\alpha E_{(b)}^\beta 
= \utnabla_b K_{ac} -  \utnabla_a K_{bc}  + K_{ac} M_b - K_{bc} M_a. 
\label{HY-Codazzi1} 
\ee

3. Choosing $X=E_{(a)}, Y=E_{(b)}$ and $\theta=\Edual^{(c)}$ in \eqref{HY-codazzi} one obtains the {\rm Codazzi-$2$ equations} 
\be
\Edual_\delta^{(d)} R^\delta_{\gamma\alpha\beta} \ell^\gamma E_{(a)}^\alpha E_{(b)}^\beta  
= \utnabla_a L_b^d -  \utnabla_b L_a^d  + L_a^d M_b - L_b^d M_a. 
\label{HY-Codazzi2} 
\ee

4. Choosing $X=E_{(a)}, Y=E_{(b)}$ and $\theta=n$ in \eqref{HY-codazzi} one obtains the {\rm Codazzi-$3$ equations} 
\be
n_\delta R^\delta_{\gamma\alpha\beta} \ell^\gamma E_{(a)}^\alpha E_{(b)}^\beta  
= \utnabla_a M_b -  \utnabla_b M_a  + K_{bc} L_a^c - K_{ac} L_b^c. 
\label{HY-Codazzi3} 
\ee
\end{theorem}

\begin{proof} The Gauss and Codazzi-$1$ equations follow directly from \eqref{HY-curvatureidentity} and \eqref{HY-movframe}. 
Now, in order to establish Codazzi-$2$ and Codazzi-$3$ equations, we write 
the equation \eqref{HY-codazzi} with $X=E_{(a)}$ and $Y=E_{(b)}$, i.e.,  
$$
\aligned
& \Riem(E_{(a)},E_{(b)})\ell 
\\
& = \utnabla_{E_{(a)}}(L_b^c E_{(c)}) - \utnabla_{E_{(b)}}(L_a^c E_{(c)}) - \left(\Gamma_{ab}^c-\Gamma_{ba}^c\right)L_c^d E_{(d)}
- M_a L_b^d E_{(d)} + M_b L_a^d E_{(d)} \\
& \quad + \left( E_{(a)}(M_b) - E_{(b)}(M_a) - \left(\Gamma_{ab}^c-\Gamma_{ba}^c\right) M_c - K_{ac}L_b^c + K_{bc}L_a^c \right) \ell.
\endaligned
$$
Since 
$$
\aligned
\utnabla_{E_{(a)}}(L_b^c E_{(c)}) & = E_{(a)}(L_b^c)\, E_{(c)} + L_b^c \Gamma_{ac}^d\, E_{(d)} 
=\left(\utnabla_{E_{(a)}} L_b^d + \Gamma_{ab}^c L_c^d\right) E_{(d)},\\
E_{(a)}(M_b)  & = \utnabla_a M_b + \Gamma_{ab}^c M_c,
\endaligned
$$
the previous equation becomes
$$
\aligned 
\Riem(E_{(a)},E_{(b)})\ell = 
& \left(\utnabla_a L_b^d -  \utnabla_b L_a^d  + L_a^d M_b - L_b^d M_a\right)E_{(d)} 
\\
& + \left( \utnabla_a M_b -  \utnabla_b M_a  + K_{bc} L_a^c - K_{ac} L_b^c\right)\ell.
\endaligned 
$$
Then the Codazzi-$2$ and Codazzi-$3$ equations are obtained by contracting this equation with $\Edual^{(d)}$ and $n$, respectively. 
\end{proof}


\section{Geometry on a hypersurface induced by a metric} 
\label{CU-0}
 
\subsection{Rigging versus normal vector fields} 
\label{ss61}

From now on, we assume that the manifold $\MM$ is endowed with a Lorentzian metric $g$. Then, $g$ induces a unique connection 
$\nabla$ on $\MM$, called the Levi-Civita connection, such that  
$$
\nabla g =0, \qquad  \Tor=0. 
$$
Our aim is now to investigate the geometry induced by $g$ and $\nabla$ on a general hypersurface $\HH\subset \MM$.

We now specialize the results in Subsection~\ref{ss52} 
to Lorentzian manifolds, and discuss the properties of the projection operators on $T\HH$ and $T^*\HH$ 
Recall the following terminology for vectors $X\in T_x\MM$: 
$$
g(X,X) \, 
\begin{cases} 
\, < 0,   & \text{timelike,}
\\ 
\, = 0,   & \text{null,}
\\ 
\, > 0,   & \text{spacelike.}
\end{cases}
$$ 
From the normal form $n$, one can determine the {\sl normal vector field} $n^\sharp$ as the unique vector field in $T\MM$ satisfying
$$
g_x(n^\sharp_x,X) = \la n_x, X \ra, \qquad X \in T_x \MM,  \quad x \in \MM.
$$
It is important to observe that supplementing the subspace $T_x\HH\subset T_x\MM$ with the normal vector $n^\sharp_x$ does 
not always yield a canonical decomposition of the tangent space $T_x\MM$, since in the Lorentzian setting the normal vector 
may well belong to $T_x\HH$ (in that case the hypersurface is called \emph{null} at $x$). Precisely, by \eqref{HY.normalform}, 
the normal vector $n^\sharp_x$ at $x \in \HH$ belongs to $T_x\HH$ if and only if $n^\sharp_x$ is a null vector, that is $g(n^\sharp,n^\sharp) = 0$.  

Therefore the prescription of a rigging vector field $\ell$ (cf. Definition \ref{HY-rigging}) is necessary for general hypersurfaces. Specifically, 

\begin{enumerate}
\item 
If $\HH$ is a general hypersurface, then $T_x\HH$, $x\in\HH$, is supplemented with a rigging vector field $\ell$ as in Section \ref{ss52}. 

\item If $\HH$ is nowhere null, then one can choose the rigging vector field to be the normal vector field, that is $\ell=n^\sharp$. 
Note that the regularity of $n^\sharp$ depends on the metric, as follows:

\item 
$n^\sharp$ and $g^{\alpha\beta}$ have the same regularity: in particular, if $g_{\alpha\beta}$ is 
uniformly non-dege\-ne\-rate and  $g_{\alpha\beta}\in L^\infty_\loc\cap W^{k,p}(\MM)$, then  
$g^{\alpha\beta}$ and $n^\sharp$ also belong to $L^\infty_\loc\cap W^{k,p}(\MM)$. 

\end{enumerate}

In Subsections~\ref{ss62}-\ref{ss64} we study the geometry of a general hypersurface $\HH$ ($\ell\neq n^\sharp$) and in Subsection~\ref{ss65} 
we specialize the results to a nowhere null hypersurface ($\ell=n^\sharp$).

 
\subsection{Metrics on $\HH$ induced by the metric $g$}
\label{ss62}

To the metric tensor $g$ (a $2$-covariant tensor) and to its inverse (a $2$-contravariant tensor),  
we associate their projections 
$$
\utg:=\utg_{ab}\, \Edual^{(a)}\otimes\Edual^{(b)}, 
\qquad
\tg:=\tg^{ab}\, E_{(a)}\otimes E_{(b)},
$$ 
whose components are given by 
$$
\utg_{ab}:=g_{\alpha\beta} \, E_{(a)}^\alpha E_{(b)}^\beta, 
\qquad 
\tg^{ab}:=g^{\alpha\beta}\Edual^{(a)}_\alpha \Edual^{(b)}_\beta.
$$  
The metric $\utg$ is simply the restriction of the original metric $g_{\alpha\beta}$ 
to the tangent space $T_x \HH$ and, as such, is independent of $\ell$, while $\ot g$ is 
the restriction of $g^{\alpha\beta}$ to the cotangent space $T^\star_x \HH$ and depends on $\ell$. 
Note that the matrix $\tg^{ab}$ is \emph{not} the inverse of the matrix $\utg_{ab}$.

Observe that the (possibly degenerate) $2$-covariant tensor $\utg_{ab}$ allows us to lower the indices of 
any vector field in $T\HH$, while the (possibly degenerate) $2$-contravariant tensor $\tg^{ab}$
allows us to raise the indices of any form in $T^\star\HH$. In particular, to the projections of a vector 
$X_x \in T_x\MM$ and a form $\theta_x \in T_x^\star \MM$ we can associate
the following form and vector, respectively, 
\be 
\tX_a := \utg_{ab} \, \tX^b, \qquad \uttheta^a := \tg^{ab} \, \uttheta_b. 
\label{HY-projection2} 
\ee
Alternatively, to the vector $X^\alpha$ and the form $\theta_\alpha$ we can first associate
the corresponding form $X_\alpha=g_{\alpha\beta} X^\beta$
and vectors $\theta^\alpha = g_{\alpha\beta} \, \theta_\beta$ and, next, project them 
to obtain the form $\utX_a$ and vector $\ot \theta^a$, respectively. 

We now investigate the properties of these projections.

\begin{theorem}[Projections of a Lorentzian metric and of vector fields] 
\label{degenerate}
Given a rigging field $\ell$, the two projections $\ut g_{ab}$ and $\ot g_{ab}$ of a Lorentzian metric 
$g:x \in \MM \to g_x \in T^0_{2,x}\MM$ satisfy the following properties: 
\begin{enumerate}
\item 
$\ut g_{ab}$ is degenerate at $x\in \HH$ if and only if $n^\sharp_x\in T_x \MM$ is a null vector.
\item 
$\ot g^{ab}$ is degenerate at $x\in \HH$ if and only if $\ell_x\in T_x \MM$ is a null vector.
\item In general, the projections of vectors do not commute with the operations of raising or lowering the indices, that is, 
for general vectors $X^\alpha$ and $1$-forms $\theta_\alpha$, 
\be 
\tX_a \neq \utX_a, \qquad \uttheta^a \neq \ttheta^a.  
\label{HY-projection3} 
\ee
\item The projections of the fields $\ell$ and $n$ satisfy 
\be
\tl^a =  0, \qquad  \utn_a = 0,  
\label{HY-projection4} 
\ee 
\be
\tn_a = -(n^\alpha n_\alpha) \, \utl_a, 
\qquad 
\tn^a \, \utl_a = 1 - (n^\alpha n_\alpha) \, (\ell^\beta \ell_\beta). 
\label{HY-projection5} 
\ee 
\end{enumerate} 
\end{theorem}

\begin{proof}  
(1) By the definition of the normal vector $n^\sharp_x\in T_x \MM$, we have 
$$
g_x(n^\sharp_x, X_x)=0, \qquad X_x\in T_x \HH.
$$
If the vector $n^\sharp_x\in T_x \MM$ is null at $x\in \HH$, then $n^\sharp_x\in T_x \HH$ 
(see Section~\ref{PG-0}) and the relation above shows that $\ut g_x$ is degenerate 
($n^\sharp_x\neq 0$ by definition). 

Conversely, if $\ut g$ is degenerate at $x\in \HH$, then there exists a vector 
$Y_x\in T_x \HH\setminus \{0\}$ such that 
$$
\ut g_x(Y_x, X_x)=0, \qquad  X_x\in T_x \HH.
$$
This implies that $\ker Y^\flat_x=T_x \HH$, where $Y^\flat_x$ denotes the form 
associated with $Y_x$. Since on the other hand the normal form satisfies $\ker n_x=T_x \HH$, 
there exists a constant $C\neq 0$ such that $Y^\flat_x=C n_x$ or, equivalently, such that 
$Y_x=C n^\sharp_x$. Hence $n^\sharp_x\in T_x \HH$, which means that $n_x$ is a null vector (see Section~\ref{PG-0}).

(2) By the definition of the dual space $T^\star_x \HH$ (see \eqref{dualspace}), the $2$-contravariant tensor $g^{-1}$ 
defined by its components $g^{\alpha\beta}$ satisfies  
$$
g^{-1}_x(\ell^\flat_x, \theta_x)=\la \theta_x,\ell_x\ra=0, 
\qquad 
 \theta_x\in T^\star_x \HH.
$$
If the vector $\ell_x\in T_x \MM$ is null at $x\in \HH$, then $\la\ell^\flat_x,\ell_x\ra=0$ 
and therefore $\ell^\flat_x\in T^\star_x \HH$ by \eqref{dualspace}. Then the relation above 
shows that $\ot g_x$, which is the restriction of $g^\sharp_x$ to $T^\star_x\HH$, is degenerate ($\ell^\flat_x\neq 0$ 
by the definition of the rigging vector).

Conversely, if $\ot g$ is degenerate at $x\in \HH$, then there exists a form $\phi_x\in T^\star_x \HH\setminus \{0\}$ such that 
$$
\ot g_x(\phi_x, \theta_x)=0, \qquad  \theta_x\in T^\star_x \HH.
$$
This implies that the kernel of $\phi^\sharp_x:T^\star_x \MM \to \RR$ is $T^\star_x \HH$ ($\phi_x$ is defined 
over the whole space $T_x\MM$ by \eqref{dualspace}). But $T^\star_x \HH$ is also the kernel of 
 $\ell_x:T^\star_x \MM \to \RR$ (see \eqref{dualspace}). 
Therefore, there exists a constant $C\neq 0$ such that $\phi^\sharp_x=C\ell_x$.  
This implies that $\ell^\flat_x\in T^\star_x \HH$, which in turn yields that $\ell_x$ is a null vector 
(see \eqref{dualspace}).

It remains to prove \eqref{HY-projection5}. Using the definitions above, we first have 
$$
\tn_a=\utg_{ab}\,\tn^b=(g_{\sigma\beta}\,E^\sigma_{(a)} E^\beta_{(b)}) (\Edual^{(b)}_\alpha n^\alpha).
$$
But, ${P_\alpha}^\beta = E^\beta_{(b)} \, \Edual^{(b)}_\alpha$, since 
$$
\aligned 
{P_\alpha}^\beta \, X^\alpha = \tX^\beta 
& = E^\beta_{(b)} \, \tX^b 
\\
& = E^\beta_{(b)} \, \Edual^{(b)}_\alpha \, X^\alpha
\endaligned 
$$
for all $X \in T\MM$. Hence, 
$$
\tn_a 
=g_{\sigma\beta}\,E^\sigma_{(a)} {P_\alpha}^\beta n^\alpha 
=g_{\sigma\beta}\,E^\sigma_{(a)} (n^\beta-\ell^\beta n_\alpha) n^\alpha
=-E^\sigma_{(a)} \ell_\sigma n_\alpha n^\alpha = -(n^\alpha n_\alpha) \, \utl_a
$$ 
and the first equation in \eqref{HY-projection5} is established. The second one is obtained 
by computing
$$
\aligned 
\tn^a \, \utl_a =\la\tn,\utl\ra
& =\tn^\alpha \, \utl_\alpha
\\
& = (n^\alpha-\ell^\alpha n_\beta n^\beta)(\ell_\alpha-n_\alpha \ell^\beta \ell_\beta)
\\
& =1 - (n^\alpha n_\alpha) \, (\ell^\beta \ell_\beta).
\endaligned
$$
\end{proof}

\subsection{Connections on $\HH$ induced by the metric $g$}  
\label{ss63}

The natural connection on $\HH$ induced by the metric $g$ is the Levi-Civita connection associated with the projected metric $\utg$ 
(the pull-back of the metric $g$). However this is not possible in general because $\utg$ is degenerate at the points where
 the hypersurface $\HH$ is null. This leads us to follow one of the following two strategies: 
\begin{enumerate}
\item 
Either define the Levi-Civita connection $\otnabla$ associated with the metric $\otg$ defined on the cotangent bundle $T^*\HH$. 

\item 
Or define the connection $\utnabla$ by projecting the connection $\nabla$ ($\utnabla$ is \emph{not} the Levi-Civita connection 
associated with the metric $\utg$ defined on the tangent bundle $T^*\HH$, save for $\ell=n^\sharp$).
\end{enumerate}

We first define the connection $\otnabla$ which requires the additional assumption that
\be
\text{the rigging vector $\ell$ is no-where null on $\HH$.}
\label{HY.assumption}
\ee
Under this assumption, Theorem~\ref{degenerate} shows that $\otg^{ab}$ is a non-degenerate tensor on $\HH$, 
and we can introduce its inverse 
$$
\big(\gamma_{ab} \big) := \big( \ot g^{ab}\big)^{-1}. 
$$
Then, the connection $\otnabla$ is defined as the unique Levi-Civita 
connection associated with the (non-degenerate) metric tensor $\gamma_{ab}$. 
This connection induced by projecting the dual of the metric $g$ will be referred to as 
the {\sl metric connection.} 

Next, we define the connection $\utnabla$ by the formula \eqref{HY-derivative}
where $\nabla$ is the Levi-Cevita connection associated with the given metric $g$. 
In particular, $\utnabla$ satisfies the properties stated in Proposition \ref{HY-connection}. 

The following Proposition gather the principal properties of the connections $\otnabla$ and $\utnabla$: 

\begin{proposition} 

1. The operator  $\otnabla$ is a metric connection (by construction) and, in particular, has zero torsion. 

2. The operator $\utnabla$ need not be a metric connection, that is, it need 
not be the Levi-Civita 
connection associated with a non-degenerate metric. 
In general, 
$$
\utnabla \utg \neq 0
$$
with the notable exception when $\ell$ is chosen to be $n^\sharp$ (for non-null hypersurfaces).  
Still, $\utnabla$ has always zero torsion, that is, 
$$
\utnabla_X Y - \utnabla_Y X - [X,Y]=0, \qquad X,Y \in \TT^1_0(\HH).  
$$

3. We have 
$$
\otnabla  = \utnabla  + F, 
$$
where $F:\TT^1_0\HH\times\TT^1_0\HH\to \Dcal'T^1_0\HH$ is the $(1,2)$-tensor field defined by 
$$
\otg^{-1}(F(X,Y),Z):=\frac12\left(
(\utnabla_X\otg^{-1})(Y,Z)+
(\utnabla_Y\otg^{-1})(X,Z)-
(\utnabla_Z\otg^{-1})(X,Y)
\right)
$$
for all $X,Y,Z\in \TT^1_0\HH$. Note that the tensor field $F$ depends on $\ell$ and that $F\neq 0$, 
except in the case that $\ell=n^\sharp$ (for non-null hypersurfaces).
\end{proposition}

\begin{proof}
We only need to prove the last assertion. Since the connection $\otnabla$ satisfies the Koszul formula ($\gamma:=\otg^{-1}$)
$$
\aligned
2 \, \gamma(\otnabla_XY,Z)=& X(\gamma(Y,Z))+Y(\gamma(X,Z))-Z(\gamma(X,Y))\\
& -\gamma(X,[Y,Z])-\gamma(Y,[X,Z])+\gamma(Z,[X,Y]),
\endaligned
$$
and the connection $\utnabla$ satisfies
$$
(\utnabla_X\gamma)(Y,Z) = X(\gamma(Y,Z)-\gamma(\utnabla_XY,Z)-\gamma(Y,\utnabla_XZ),
$$
we deduce that 
$$
\aligned
2 \, \gamma(\otnabla_XY,Z)
=& 
(\utnabla_X\gamma)(Y,Z)+(\utnabla_Y\gamma)(X,Z)-(\utnabla_Z\gamma)(X,Y) + 2\,\gamma(\utnabla_XY,Z)\\
& 
-\gamma(\utnabla_XY-\utnabla_YX-[X,Y],Z)
+\gamma(\utnabla_XZ-\utnabla_ZX-[X,Z],Y)\\
&
+\gamma(\utnabla_YZ-\utnabla_ZY-[Y,Z],X).
\endaligned
$$
But the connection $\utnabla$ has zero torsion and thus the above formula reduces to 
$$
2 \, \gamma(\otnabla_XY,Z)=2\gamma(F(X,Y),Z)+2\,\gamma(\utnabla_XY,Z),
$$
where $F$ is the tensor defined in the statement of the Proposition. The proof is completed.
\end{proof}

\subsection{Gauss and Codazzi equations} 
\label{ss64}
The Riemann curvature tensors defined on the hypersurface $\HH$ by the connections $\utnabla$ and $\otnabla$ are related 
to one another via the tensor $F$ and its covariant derivative $\nabla F$, since $\otnabla  = \utnabla  + F$. Moreover, 
we have seen in Subsection~\ref{ss56} that the Riemann curvature tensor associated with $\utnabla$ and the Riemann curvature tensor 
associated with $\nabla$ are related by the Gauss and Codazzi equations given by Proposition~\ref{55} 
(the assumptions of this Proposition are clearly satisfied). 

We now take advantage of the fact that $\nabla$ is a metric connection (i.e., $\nabla$ is the Levi-Civita connection 
associated with the metric $g$) and establish further properties of the geometry of $\HH$. It is well known that 
the $4$-covariant Riemann tensor defined by 
$$
\Riem(W,Z,X,Y) := g(W,\Riem(X,Y)Z) 
$$
satisfies the symmetries
\be
\aligned 
\Riem(W,Z,X,Y) & = \Riem(X,Y,W,Z)
\\
& = - \Riem(W,Z,Y,X) = -\Riem(Z,W,X,Y), 
\endaligned 
\label{HY-symmetries} 
\ee
as well as the Bianchi identities
\be
\aligned
\Riem(W,Z,X,Y) + \Riem(W,X,Y,Z) + \Riem(W,Y,Z,X) & =0,\\
\nabla_X\Riem(W,U,Y,Z) + \nabla_Y\Riem(W,U,Z,X) + \nabla_Z\Riem(W,U,X,Y) & =0,
\endaligned
\label{HY-bianchi1}
\ee
for all $X,Y,Z,U,W\in \TT^1_0(\MM)$. This implies 
that the left-hand sides of the Gauss and Codazzi equations satisfies the above symmetry relations. 
Consequently, their right-hand sides must also satisfy these symmetry relations. 
In view of the relations \eqref{HY-Gauss}-\eqref{HY-Codazzi3}, this yields 
the compatibility relations that $\undertilde R_{abc}^d$, $K_{ab}$, $L_a^c$ and $M_a$ must satisfy.


\subsection{The case of nowhere null hypersurfaces}
\label{ss65}

If the hypersurface $\HH$ is no-where null, i.e., the normal vector $n_x^\sharp$ is not null at 
any point of $x \in \HH$, then we can choose the particular rigging vector 
$$
\ell=n^\sharp.
$$ 
With this choice, the following properties hold: 

\begin{enumerate}
\item The projection operators do commute with the operations of raising or lowering the indices, that is, 
for general vectors $X^\alpha$ and $1$-forms $\theta_\alpha$, 
$$
\tX_a = \utX_a, \qquad \uttheta^a = \ttheta^a.  
$$
\item The two connections defined in Section \ref{ss63} coincide, 
$$
\ut \nabla = \ot \nabla, 
$$
and are nothing but the Levi-Civita connection associated with the metric $\ut g_{ab}$. 
\item The coefficients appearing in the equations \eqref{HY-movframe} satisfy 
$$
L_a^c=K_{ab} g^{bc}, \qquad  M_a=0.
$$
This is consequence of the relations
$$
\aligned
0 & =\nabla_{E_{(a)}}(g(E_{(b)},n^\sharp))=g(\nabla_{E_{(a)}}E_{(b)},n^\sharp)+g(E_{(b)},\nabla_{E_{(a)}} n^\sharp)\\
& =-K_{ab}+\ut g_{bc} L_a^c
\endaligned
$$
and 
$$
0=\nabla_{E_{(a)}}(g(n^\sharp,n^\sharp))=2g(\nabla_{E_{(a)}}n^\sharp,n^\sharp)=2M_a, 
$$
themselves following from the fact that the connection on $\MM$ 
satisfies $\nabla g=0$ combined with the orthogonality between the normal vector $n^\sharp$ and the vectors fields 
$E_{(c)}\in T\HH$ (i.e., $g(n^\sharp,E_{(c)})=0$). 
 
\item The Gauss and Codazzi equations \eqref{HY-Gauss}--\eqref{HY-Codazzi3} reduce to the usual Gauss and 
Codazzi equations associated with a hypersurface. 
Indeed, using $\ell=n^\sharp$ one can see 
 that the Codazzi-$2$ equations are equivalent to the 
 Codazzi-$1$ equations and that the Codazzi-$3$ equations vanish identically.
\end{enumerate}


\section*{Acknowledgments}  

The first author (PLF) is very thankful to J.M. Stewart (DAMTP, Cambridge University) for 
many enlightening discussions
on the subject of this paper, as well as to 
the organizers (P.T. Chrusciel, H. Friedrichs, P. Tod) of the Semester Program
``Global Problems in Mathematical Relativity'' which took place at the Isaac Newton Institute of Mathematical Sciences 
(Cambridge, UK) where this research was initiated. 
This research was partially supported by the A.N.R. (Agence Nationale de la Recherche)
through the grant 06-2-134423 entitled {\sl ``Mathematical Methods in General Relativity''} (MATH-GR), 
and by the Centre National de la Recherche Scientifique (CNRS).



\begin{thebibliography}{99}

\newcommand{\auth}{\textsc} 

\bibitem{BLSS} \auth{Barnes A.P., LeFloch P.G., Schmidt B.G., and Stewart J.M.,}
The Glimm scheme for perfect fluids on plane-symmetric Gowdy spacetimes, 
{\em Class. Quantum Grav.} {\bf 21} (2004), 5043--5074.

\bibitem{BarrabesIsrael} \auth{Barrab\`es C. and Israel W.,} 
Thin shells in general relativity and cosmology: the lightlike limit, 
{\em Phys. Rev. D} {\bf 43} (1991), 1129--1142.  

\bibitem{CaciottaNicolo} \auth{Caciotta G. and Nicolo F.,} 
Global characteristic problem for Einstein vacuum equations with small initial data: the initial data constraints,
{\em J. Hyper. Differ. Equa.} {\bf 2} (2005), 201--248. 

\bibitem{ClarkeDray} \auth{Clarke C.J.S. and Dray T.,}
Junction conditions in null hypersurfaces, 
{\em Class. Quantum Grav.} {\bf 4} (1987), 265--275. 

\bibitem{FriedrichNagy} \auth{Friedrich H. and Nagy G.}
The initial boundary value problem for Einstein's vacuum field equations, 
{\em Comm. Math. Phys.} {\bf 201} (1999), 619--655.

\bibitem{FriedrichRendall} \auth{Friedrich H. and Rendall A.D.,}  
The Cauchy problem for the Einstein equations, 
in ``Einstein's Field Equations and their Physical Interpretation'', 
Ed. B.G. Schmidt, Springer Verlag, 2000, pp.~127--223.   

\bibitem{GerochTraschen} \auth{Geroch R.P. and Traschen J.,}
Strings and other distributional sources in general relativity, 
{\em Phys. Rev. D} {\bf 36} (1987) 1017--1031. 

\bibitem{LeFlochStewart} \auth{LeFloch P.G. and Stewart J.M.,}
Shock waves and gravitational waves in matter spacetimes with Gowdy symmetry, 
{\em Portugal. Math.} {\bf 62} (2005), 349--370. 

\bibitem{Lichne} \auth{Lichnerowicz A.,}
{\sl Magnetohydrodynamics: waves and shock waves in curved spacetime,}
Kluwer Acad. Publisher, Vol. {\bf 14,} 1993. 

\bibitem{MarsSenovilla} \auth{Mars M. and Senovilla J.M.,}
Geometry of general hypersurfaces in spacetime: junction conditions,  
{\em Class. Quantum Grav.} {\bf 10} (1993), 1865--1897. 

\bibitem{Parker} \auth{Parker P.E.,}  
Distributional geometry, 
{\em J. Math. Phys.} {\bf 20} (1979), 1423--1426.  

\bibitem{Penrose} \auth{Penrose R.,} 
The geometry of impulsive gravitational waves, in ``General Relativity'', 
Papers in honor of J.L. Synge, edited by L. O'Raifeartaigh, Clarendon Press, Oxford, 1972, pp.~101--115. 

\bibitem{Raju} \auth{Raju C.K.,}  
Distributional matter tensors in relativity, 
Proc. Fifth Marcel Grossmann Meeting on General Relativity, 
Part A, B (Perth, 1988), World Sci. Publishing, Teaneck, NJ, 1989, pp.~419--422. 

\bibitem{Rendall} \auth{Rendall A.D.,} 
Reduction of the characteristic initial value problem to the Cauchy problem and its applications
to the Einstein equations, 
{\em Proc. Royal Soc. Lond. A} {\bf 427} (1990), 221--239. 

\bibitem{Rendall2} \auth{Rendall A.D.,}
Blow-up for solutions of hyperbolic PDE and spacetime singularities,
Journ\'ees Equations aux D\'eriv\'ees Partielles, Nantes, June 5--9,
2000, GRD 1151, CNRS, France. 

\bibitem{ReulaSarbach} \auth{Reula O. and Sarbach O.,}
A model problem for the initial boundary value problem formulation of Einstein's field equations.
{\em J. Hyper. Differ. Equa.} {\bf 2} (2005), 397--436. 
		
\bibitem{SmollerTemple} \auth{ Smoller J. and Temple B.,} 
Multi-dimensional shock-waves for relativistic fluids,
in ``Nonlinear evolutionary partial differential equations'' (Beijing, 1993), 
AMS/IP Stud. Adv. Math., Vol. {\bf 3}, Amer. Math. Soc., Providence, RI, 1997, pp.~377--391.

\bibitem{Stewart} \auth{Stewart J.M.,} 
The Cauchy problem and the initial-boundary value problem in numerical relativity, 
{\em Class. Quantum Grav.} {\bf 15} (1998) 2865--2889.

\bibitem{Taub} \auth{Taub A.H.,} 
Spacetimes with distribution-valued curvature tensors, 
{\em J. Math. Phys.} {\bf 21} (1980), 1423--1431. 

\bibitem{Wald} \auth{Wald R.M.,} 
{\sl General relativity,} University of Chicago Press, 1984. 

\end{thebibliography}
\end{document}